\documentclass[12pt,onecolumn,draftclsnofoot]{IEEEtran} 

\usepackage{url}
\usepackage{tikz}
\usepackage{pgf}
\usepackage{amsthm,amssymb,amsmath}
\usepackage{cite,url}
\usepackage{graphicx,subfigure}
\usepackage{times,verbatim,color,stackrel,bm}

\newtheorem{theorem}{Theorem}

\newtheorem{definition}{Definition}
\newtheorem{subsec:coding}{subsec:coding}

\newtheorem{lemma}{Lemma}

\begin{document}

\title{On Cooperative Relay Networks with Video Applications}

\author{
\authorblockN{Donglin Hu \ \ and \ \ Shiwen Mao \\}
\authorblockA{Department of Electrical and Computer Engineering  \\
Auburn University, Auburn, AL, USA} 
}

\maketitle

\pagestyle{plain}\thispagestyle{plain}


\begin{abstract}
In this paper, we investigate the problem of cooperative relay in CR networks for further enhanced network performance. In particular, we focus on the two representative cooperative relay strategies, and develop optimal spectrum sensing and $p$-Persistent CSMA for spectrum access. Then, we study the problem of cooperative relay in CR networks for video streaming. We incorporate interference alignment to allow transmitters collaboratively send encoded signals to all CR users. In the cases of a single licensed channel and multiple licensed channels with channel bonding, we develop an optimal distributed algorithm with proven convergence and convergence speed. In the case of multiple channels without channel bonding, we develop a greedy algorithm with bounded performance.
\end{abstract}

\section{Introduction}
{\em Cooperative relay} in CR networks~\cite{Simeone07, Zhang09} represents another new paradigm for wireless communications. It allows wireless CR nodes to assist each other in data delivery, with the objective of achieving greater reliability and efficiency than each of them could attain individually (i.e., to achieve the so-called {\em cooperative diversity}). Cooperation among CR nodes enables opportunistic use of energy and bandwidth resources in wireless networks, and can deliver many salient advantages over conventional point-to-point wireless communications.

Recently, there has been some interesting work on cooperative relay in CR networks~\cite{Hu10GC,Simeone07, Zhang09}.  In~\cite{Simeone07}, the authors considered the case of two single-user links, one primary and one secondary.  The secondary transmitter is allowed to act as a ``transparent'' relay for the primary link, motivated by the rationale that helping primary users 
will lead to more transmission opportunities for CR nodes.  In~\cite{Zhang09}, the authors presented an excellent overview of several cooperative relay scenarios and various related issues. A new MAC protocol was proposed and implemented in a testbed to select a spectrum-rich CR node as relay for a CR transmitter/receiver pair. 

We investigate cooperative relay in CR networks, using video as a reference application to make the best use of the enhanced network capacity~\cite{Hu12IC}. We consider a base station (BS) and multiple relay nodes (RN) that collaboratively stream multiple videos to CR users within the network. To support high quality video service in such a challenging environment, we assume a well planned relay network where the RNs are connected to the BS with high-speed wireline links. Therefore the video packets will be available at both the BS and the RNs before their scheduled transmission time, thus allowing advanced cooperative transmission techniques to be adopted for streaming videos. In particular, we consider interference alignment, where the BS and RNs simultaneously transmit encoded signals to all CR users, such that undesired signals will be canceled and the desired signal can be decoded at each CR user~\cite{Tse05, Cadambe08}. In~\cite{Li10}, such cooperative sender-side techniques are termed {\em interference alignment}, while receiver-side techniques that use overheard (or exchanged via a wireline link) packets to cancel interference is termed {\em interference cancelllation}. We present a stochastic programming formulation of the problem of interference alignment for video streaming in cooperative CR networks and then a reformulation of the problem based on Linear Algebra theory~\cite{Strang09}, such that the number of variables and computational complexity can be greatly reduced. To address the formulated problem, we propose an optimal distributed algorithm with proven convergence and convergence rate, and then a greedy algorithm with a proven performance bound.

The remainder of this paper is organized as follows. Related work is discussed in Section~\ref{sec:coop_work}. In Section~\ref{sec:cr_relay}, we compare two cooperative relay strategies in CR networks. We investigate the problem of cooperative CR relay with interference alignment for MGS video streaming in Section~\ref{sec:cr_video_coop}. Section~\ref{sec:coop_conc} concludes the paper.

\section{Background and Related Work}\label{sec:coop_work}
The theoretical foundation of relay channels was laid by the seminal work~\cite{Cover79}. 
The capacities of the Gaussian relay channel and certain discrete relay channels are evaluated, and the achievable lower bound to the capacity of the general relay channel is established in this work.  In~\cite{Sendonaris03a, Sendonaris03b}, the authors described the concept of cooperative diversity,
where diversity gains are achieved via the cooperation of mobile users. 
In~\cite{Laneman04}, the authors developed and analyzed low-complexity cooperative
diversity protocols. Several cooperative strategies, including AF and DF, were described and their performance characterizations were derived in terms of outage probabilities.

In practice, there is a restriction that each node cannot transmit and receive simultaneously in the same frequency band. The ``cheap'' relay channel concept was introduced in~\cite{Khojastepour03}, where the authors 
derived the capacity of the Gaussian degraded ``cheap'' relay channel. Multiple relay nodes for a transmitter-receiver pair are investigated in~\cite{Zhao07b} and~\cite{Bletsas06}. 
The authors showed that, when compared with complex protocols that involve all relays, the simplified protocol with no more than one relay chosen can achieve the same performance. This is the reason why we consider single relay in this paper. 

In~\cite{Ng07}, Ng and Yu proposed a utility maximization framework for joint optimization of node, relay strategy selection, and power, bandwidth and rate allocation in a cellular network. 
Cai et al.~\cite{Cai08} presented a semi-distributed algorithm for AF relay networks. A heuristic was adopted to select relay and allocate power.  Both AF and DF were considered in~\cite{Shi08}, where a polynomial time algorithm for optimal relay selection was developed and proved to be optimal.
In~\cite{Ding10}, a protocol is proposed for joint routing, relay selection, and dynamic spectrum allocation for multi-hop CR networks, and its performance is evaluated through simulations. 

The problem of video over CR networks has only been studied in a few recent papers~\cite{Hu10JSAC, Shiang08, Ding09, Hu10TW, Luo11}.
In~\cite{Shiang08}, a dynamic channel selection scheme was proposed for CR users to transmit videos over multiple channels. In~\cite{Ding09}, a distributed joint routing and spectrum sharing algorithm for video streaming over CR ad hoc networks was described and evaluated with simulations. In our prior work, we considered video multicast in an infrastructure-based CR network~\cite{Hu10JSAC}, unicast video streaming over multihop CR networks~\cite{Hu10TW} and CR femtocell networks~\cite{Hu12JSAC}.  
In~\cite{Luo11}, the impact of system parameters residing in different network layers are jointly considered to achieve the best possible video quality for CR users. 
Unlike the heuristic approaches in~\cite{Shiang08, Ding09}, the analytical and optimization approach taken in this paper yields algorithms with optimal or bounded performance. The cooperative relay and interference alignment techniques also distinguish this paper from prior work on this topic. 

As point-to-point link capacity approaches the Shannon limit, there has been considerable interest on exploiting interference to improve wireless network capacity~\cite{Tse05, Cadambe08, Katti07, Gollakota09, Li10}. In addition to information theoretic work on asymptotic capacity~\cite{Tse05, Cadambe08}, practical issues have been addressed in~\cite{Katti07, Gollakota09, Li10}.
In~\cite{Katti07}, the authors presented a practical design of analog network coding to exploit interference and allow concurrent transmissions, which does not make any synchronization assumptions. 
In~\cite{Gollakota09}, interference alignment and cancellation is incorporated in MIMO LANs, and the network capacity is shown, analytically and experimentally, to be almost doubled. 
In~\cite{Li10}, the authors presented a general algorithm for identifying interference alignment and cancellation opportunities in practical multi-hop mesh networks. The impact of synchronization and channel estimation was evaluated through a GNU Radio implementation. Our work was motivated by these interesting papers, and we incorporate interference alignment in cooperative CR networks and exploit the enhanced capacity for wireless video streaming.

\section{CR and Cooperative Networking}\label{sec:cr_relay}
In this section, we investigate the problem of cooperative relay in CR networks~\cite{Hu10GC}.  We assume a primary network with multiple licensed bands and a CR network consisting of multiple cooperative relay links.  Each cooperative relay link consists of a CR transmitter, a CR relay, and a CR receiver. 
The objective is to develop effective mechanisms to integrate these two wireless communication technologies, and to provide an analysis for the comparison of two representative cooperative relay strategies, i.e., {\em decode-and-forward} (DF) and {\em amplify-and-forward} (AF), in the context of CR networks.  We first consider cooperative spectrum sensing by the CR nodes. We model both types of sensing errors, i.e., miss detection and false alarm, and derive the optimal value for the sensing threshold.  Next, we incorporate DF and AF into the $p$-Persistent Carrier Sense Multiple Access (CSMA) protocol for channel access for the CR nodes.  We develop closed-form expressions for the network-wide capacities achieved by DF and AF, respectively, as well as that for the case of direct link transmission for comparison purpose. 

Through analytical and simulation evaluations of DF and AF-based cooperative relay strategies, we find the analysis provides upper bounds for the simulated results, which are reasonably tight.  We also find cross-point with the AF and DF curves when some system parameter is varied, indicating that each of them performs better in a certain parameter range.  There is no case that one completely dominates the other for the two strategies.  The considerable gaps between the cooperative relay results and the direct link results exemplify the diversity gain achieved by cooperative relays in CR networks. 

\subsection{Network Model and Assumptions \label{sec:mod}}

We assume a primary network and a spectrum band that is divided into $M$ licensed channels, 
each modeled as a time slotted, block-fading channel.  The state of each channel evolves independently following a discrete time Markov process.

As illustrated in Fig.~\ref{fig:network}, there is a CR network colocated with the primary network.  The CR network consists of $N$ sets of cooperative relay links, each including a CR transmitter, a CR relay, and a CR receiver. Each CR node (or, secondary user) is equipped with two transceivers, each incorporating a software defined radio (SDR) that is able to tune to any of the $M$ licensed channels and a control channel and operate from there.  

\begin{figure} [!t]
\centering
\includegraphics[width=4.5in, height=2.7in]{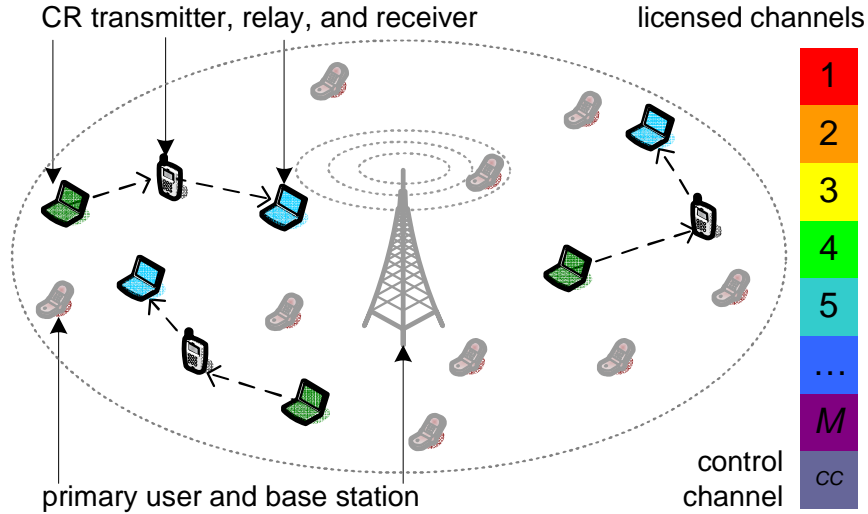}
\caption{Illustration of colocated primary and CR networks. The CR network consists of a number cooperative relay links, each consisting of a CR transmitter, a CR relay and a CR receiver.}
\label{fig:network}
\end{figure}

We assume CR nodes access the licensed channels following the same time slot structure~\cite{Zhao07a}.   In the sensing phase, a CR node chooses one of the $M$ channels to sense using one of its transceivers, and then exchanges sensed channel information  with other CR nodes using the other transceiver over the control channel.  During the transmission phase, the CR transmitter and/or relay 
transmit data frames on licensed channels that are believed to be idle based on sensing results, using one or both of the transceivers. 
We consider cooperative relay strategies AF and DF, and compare their performance in the following sections.


\subsection{Cooperative Relay in CR Networks \label{sec:ana_relay}}

In this section, we investigate how to effectively integrate the two advance wireless communication technologies, and present an analysis of the cooperative relay strategies in CR networks. We first examine cooperative spectrum sensing and 
derive the optimal sensing threshold.  We then consider cooperative relay and spectrum access, and derive the network-wide throughput performance achievable when these two technologies are integrated.

\subsubsection{Spectrum Sensing \label{subsec:sen}}
We assume there are $N_m$ CR nodes sensing channel $m$. After the sensing phase, each CR node obtains a {\em sensing result vector} 
$\vec{\Theta}_m=[\Theta_1^m, \Theta_2^m, \cdots, \Theta_{N_m}^m]$ for channel $m$.  The conditional probability 
$a_m(\vec{\Theta}_m)$ on channel $m$ availability is
\begin{eqnarray}\label{eq:CondAvail_relay}
&& a_m(\Theta_1^m,\Theta_2^m,\cdots,\Theta_{N_m}^m) \nonumber \\
&\cong& \Pr\{H_0^m|\Theta_1^m,\Theta_2^m,\cdots,\Theta_{N_m}^m\} \nonumber \\
&=&\frac{\Pr\{\Theta_1^m,\Theta_2^m,\cdots,\Theta_{N_m}^m|H_0^m)\}\Pr\{H_0^m\}} {\sum_{j\in\{0,1\}} \Pr\{\Theta_1^m,\Theta_2^m,\cdots,\Theta_{N_m}^m|H_j^m\} \Pr\{H_j^m\}} \nonumber \\
&=&\frac{\prod_{i=1}^{N_m}\Pr\{\Theta_i^m|H_0^m\} \Pr\{H_0^m\}}{\sum_{j\in\{0,1\}}\prod_{i=1}^{N_m}\Pr\{\Theta_i^m|H_j^m\} \Pr\{H_j^m\}} \nonumber \\
&=&\left[ 1+\frac{\Pr\{H_1^m\}}{\Pr\{H_0^m\}}\prod_{i=1}^{N_m}\frac{\Pr\{\Theta_i^m|H_1^m\}}{\Pr\{\Theta_i^m|H_0^m\}} \right]^{-1} \nonumber \\
&=&\left[1+\frac{\eta_m}{1-\eta_m}\prod_{i=1}^{N_m}\frac{(\delta_i^m)^{1-\Theta_i^m}(1-\delta_i^m)^{\Theta_i^m}}{(\epsilon_i^m)^{\Theta_i^m}(1-\epsilon_i^m)^{1-\Theta_i^m}}\right]^{-1}.
\end{eqnarray}
If $a_m(\vec{\Theta}_m)$ is greater than a {\em sensing threshold} $\tau_m$, channel $m$ is believed to be idle; otherwise, channel $m$ is believed to be busy. The decision variable $D_m$ is defined as follows.
\begin{equation}\label{eq:DefDeci}
D_m=\left\{\begin{array}{ll}
0, & \mbox{if } a_m(\vec{\Theta}_m)>\tau_m \\
1, & \mbox{if } a_m(\vec{\Theta}_m)\leq\tau_m.
\end{array}\right.
\end{equation}

CR nodes only attempt to access channel $m$ where $D_m$ is $0$. Since function $a_m(\vec{\Theta}_m)$ in (\ref{eq:CondAvail_relay}) has $N_m$ binary variables, there can be $2^{N_m}$ different combinations corresponding to $2^{N_m}$ values for $a_m(\vec{\Theta}_m)$. We sort the $2^{N_m}$ combinations according to their $a_m(\vec{\Theta}_m)$ values in the non-increasing order.  Let $a_m^{(j)}$ be the $j$th largest function value and $\vec{\theta}_m^{(j)}$ the argument that achieves the $j$th largest function value $a_m^{(j)}$, where 
$$
\vec{\theta}_m^{(j)}=[\theta_1^m(j),\theta_2^m(j),\cdots,\theta_{N_m}^m(j)].
$$

In the design of CR networks, we consider two objectives: (i) how to avoid harmful interference to primary users, and (ii) how to fully exploit spectrum opportunities for the CR nodes.  For primary user protection, we limit the collision probability with primary user with a threshold. 
Let $\gamma_m$ be the {\em tolerance threshold}, i.e., the maximum allowable interference probability with primary users on channel $m$. The probability of collision with primary users on channel $m$ is given as $\Pr\{D_m=0\:|\:H_1^m\}$; the probability of detecting an available transmission opportunity is $\Pr\{D_m=0\:|\:H_0^m\}$. Our objective is to maximize the probability of detecting available channels, while keeping the collision probability below $\gamma_m$. Therefore, the optimal spectrum sensing problem can be formulated as follows.  
\begin{eqnarray}\label{eq:SenseOpt}
\max_{\tau_m} & & \Pr\{D_m=0|H_0^m\} \label{eq:SenseOptObj} \\
\mbox{subect to:} & & \Pr\{D_m=0|H_1^m\} \le \gamma_m. \label{eq:SenseOptCnt}
\end{eqnarray}

From their definitions, both $\Pr\{D_m=0\:|\:H_1^m\}$ and $\Pr\{D_m=0\:|\:H_0^m\}$ are decreasing functions of $\tau_m$. 
As $\Pr\{D_m=0\:|\:H_1^m)\}$ approaches its maximum allowed value $\gamma_m$, $\Pr\{D_m=0\:|\:H_0^m\}$ also approaches its maximum.  Therefore, solving the optimization problem (\ref{eq:SenseOptObj}) $\sim$ (\ref{eq:SenseOptCnt}) is equivalent to solving 
$$
\Pr\{D_m=0\:|\:H_1^m\}=\gamma_m.  
$$ 
If $\tau_m = a_m^{(j)}$, we have 
\begin{eqnarray}
&& \Pr\{D_m=0|H_1^m\}(a_m^{(j)}) 
= \Pr\{a_m (\vec{\Theta}_m) > a_m^{(j)}|H_1^m\}  \nonumber \\
&=& \sum_{l=1}^{j-1} \Pr\{a_m (\vec{\Theta}_m) = a_m^{(l)}|H_1^m\} 
=\sum_{l=1}^{j-1} (\delta_i^m)^{1-\theta_i^m(l)}(1-\delta_i^m)^{\theta_i^m(l)}. \label{eq:PrD0H1}
\end{eqnarray}
Obviously, $\Pr\{D_m=0\;|\;H_1^m\}(a_m^{(j)})$ is an increasing function of $j$. 
The optimal sensing threshold $\tau_m^{*}$ can be set to $a_m^{(j)}$, such that 
$$
\Pr\{D_m=0\:|\:H_1^m\}(a_m^{(j)}) \leq \gamma_m
$$ 
and 
$$
\Pr\{D_m=0\:|\:H_1^m\}(a_m^{(j+1)}) > \gamma_m.
$$
The algorithm for computing the optimal sensing threshold $\tau_m^{*}$ is presented in Table~\ref{tab:sensing}. 

\begin{table} [!t]
\begin{center}
\caption{Algorithm for Computing the Optimal Sensing Threshold}
\begin{tabular}{ll}
\hline
1:  & Compute $a_m^{(j)}$ and the corresponding $\vec{\theta}_m^{(j)}$, \\
    & for all $j$; \\ 
2:  & Initialize $p_c = \Pr\{a_m (\vec{\Theta}_m)=a_m^{(1)}|H_1^m \}$ and \\
    & $\tau_m=a_m^{(1)}$; \\
3:  & Set $j=1$; \\
4:  & {WHILE} ($p_c \leq \gamma_m$) \\
5:  & $\;\;$ $j=j+1$; \\
6:  & $\;\;$ $\tau_m=a_m^{(j)}$; \\
7:  & $\;\;$ $p_c = p_c + \Pr\{a_m(\vec{\Theta}_m)=a_m^{(j)}|H_1^m\}$; \\
8:  & {END WHILE} \\
\hline
\end{tabular}
\label{tab:sensing}
\end{center}
\end{table}

Once the optimal sensing threshold $\tau_m^{*}$ 
is determined, $\Pr\{D_m=0\:|\:H_1^m\}$ can be computed as given in (\ref{eq:PrD0H1}) and $\Pr\{D_m=0\:|\:H_0^m\}$ can be computed as:
\begin{eqnarray}
&&  \Pr\{D_m=0|H_0^m\} 
= \Pr\{a_m (\vec{\Theta}_m) > \tau_m^{*}|H_0^m\} \nonumber \\
&=& \sum_{l=1}^{j-1} \Pr\{a_m (\vec{\Theta}_m) = a_m^{(l)}|H_0^m\} 
= \sum_{l=1}^{j-1} (\epsilon_i^m)^{\theta_i^m(l)}(1-\epsilon_i^m)^{1-\theta_i^m(l)}. \label{eq:PrD0H0}
\end{eqnarray}

\subsubsection{Cooperative Relay Strategies \label{subsec:coop-relay}}

During the transmission phase, CR transmitters and relays attempt to send data through the channels that are believed to be idle.  
We assume fixed length for all the data frames. 
Let $G_1^k$ and $G_2^k$ denote the path gains from the transmitter to relay and from the relay to receiver, respectively, 
and let $\sigma_{r,k}^2$ and $\sigma_{d,k}^2$ denote the noise powers at the relay and receiver, respectively, for the $k$th cooperative relay link. We examine the two cooperation relay strategies DF and AF in the following. For comparison purpose, we also consider direct link transmission below.

\begin{figure} [!t]
\centering
\includegraphics[width=5.0in]{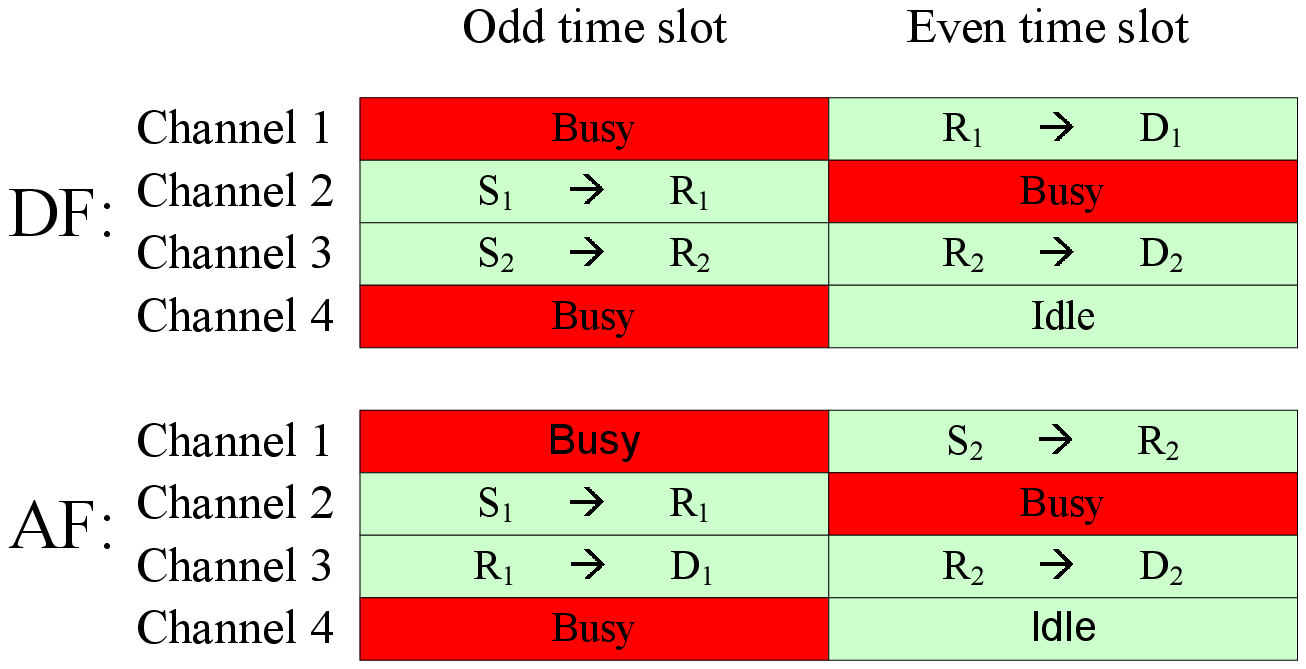}
\caption{Illustration of the protocol operation of AF and DF, where $S_i \Rightarrow R_i$ represents the transmission from source to relay and $R_i \Rightarrow D_i$ represents the transmission from relay to destination, for the $i$th cooperative relay link.}
\label{fig:protocol}
\end{figure}

\paragraph{Decode-and-Forward (DF)}

With DF, the CR transmitter and relay transmit separately on consecutive odd and event time slots: the CR transmitter sends data to the corresponding relay in an {\em odd} time slot; the relay node then decodes the data and forwards it to the receiver in the following {\em even} time slot, as shown in Fig.~\ref{fig:protocol}.  

Without loss of generality, we assume a data frame can be successfully decoded if the received signal-to-noise ratio (SNR) is no less than a {\em decoding threshold} $\kappa$. We assume gains on different links are independent to each other.  
The receiver can successfully decode the frame if it is not lost or corrupted on both links.  
The {\em decoding rate} of DF at the $k$th receiver, denoted by $P_{DF}^k$, can be computed as,
\begin{eqnarray}\label{eq:PrDF}
%
P_{DF}^k \hspace{-0.025in} &=& \hspace{-0.025in} \Pr\left\{ \left( P_sG_1^k / \sigma_{r,k}^2 \geq \kappa \right) \mbox{ and } \left( P_rG_2^k / \sigma_{d,k}^2 \geq \kappa \right) \right\} \nonumber \\
%
&=&\bar{F}_{G_1^k} \left( \sigma_{r,k}^2 \kappa / P_s \right) \bar{F}_{G_2^k} \left( \sigma_{d,k}^2 \kappa / P_r \right),
\end{eqnarray}
where $P_s$ and $P_r$ are the transmit powers at the transmitter and relay, respectively, 
$\bar{F}_{G_1^k}(x)$ and $\bar{F}_{G_2^k}(x)$ are the complementary cumulative distribution functions (CCDF) of path gains $G_1^k$ and $G_2^k$, respectively.

\paragraph{Amplify-and-Forward (AF)}

With AF, the CR transmitter and relay transmit simultaneously in the same time slot on different channels. A pipeline is formed connecting the CR transmitter to the relay and then to the receiver; the relay amplifies the received signal and immediately forwards it to the receiver in the same time slot, as shown in Fig.~\ref{fig:protocol}. Recall that the CR relay has two transceivers. The relay receives data from the transmitter using one transceiver operating on one or more idle channels; it forwards the data simultaneously to the receiver using the other transceiver operating on one or more {\em different} idle channels.  

With this cooperative relay strategy, a data frame can be successfully decoded if the SNR at the receiver is no less than the decoding threshold $\kappa$. Then the decoding rate of AF at the $k$th receiver, denoted as $P_{AF}^k$, can be computed as,
\begin{eqnarray}\label{eq:PrAF}
P_{AF}^k&=&
\Pr\left\{ \frac{P_r}{G_1^kP_s+\sigma_{r,k}^2}\frac{P_sG_1^kG_2^k}{\sigma_{d,k}^2}\geq \kappa \right\} 
= \int_0^{+\infty} \hspace{-0.03in} \bar{F}_{G_2^k}\left( \frac{(P_sx+\sigma_{r,k}^2)\sigma_{d,k}^2\kappa}{P_s P_rx} \right) dF_{G_1^k}(x). \nonumber \\
\end{eqnarray}

\paragraph{Direct Link Transmission}

For comparison purpose, we also consider the case of direct link transmission (DL). That is, the CR transmitter transmits to the receiver via the direct link; the CR relay is not used in this case.  Let the path gain be $G_0^k$ with CCDF $\bar{F}_{G_0^k}(x)$, and recall that the noise power is $\sigma_{d,k}^2$ at the receiver, for the $k$th direct link transmission.  

Following similar analysis, the decoding rate of DL at the $k$th receiver, denoted as $P_{DL}^k$, can be computed as
\begin{eqnarray} 
P_{DL}^k &=& \Pr\left\{ P_sG_0^k / \sigma_{d,k}^2 \geq \kappa \right\} 
= \bar{F}_{G_0^k} \left( \sigma_{d,k}^2\kappa / P_s \right). \label{eq:PrDL}
\end{eqnarray}

\subsubsection{Opportunistic Channel Access \label{subsec:channaccess}}

We assume greedy transmitters that always have data to send. The CR nodes
use $p$-Persistent CSMA
for channel access.  At the beginning of the transmission phase of an odd time slot, CR transmitters send Request-to-Send (RTS) with probability $p$ over the control channel. 
Since there are $N$ CR transmitters, the transmission probability $p$ is set to $1/N$ to maximize the throughput (i.e., to maximize $P_1$ in (\ref{eq:PrCSMA}) given below).  

The following three cases may occur: 
\begin{itemize}
  \item {\em Case 1}: none of the CR transmitters sends RTS for channel access. The idle licensed channels will be wasted.  
  \item {\em Case 2}: only one CR transmitter sends RTS, and it successfully receives Clear-to-Send (CTS) from the receiver over the control channel. It then accesses some of or all the licensed channels that are believed to be idle for data transmission in the transmission phase. 
  \item {\em Case 3}: more than one CR transmitters send RTS and collision occurs on the control channel. No CR node can access the licensed channels, and the idle licensed channels will be wasted. 
\end{itemize}

Let $P_0$, $P_1$ and $P_2$ denote the probability corresponding to the three cases enumerated above, respectively.  We then have
\begin{eqnarray}\label{eq:PrCSMA}
P_0 &=& (1-p)^N \;=\; (1 - 1/N)^N \\
P_1 &=& Np(1-p)^{N-1} \;=\; (1 - 1/N)^{N-1} \\
P_2 &=& 1-P_0-P_1.
\end{eqnarray}      
The CR cooperative relay link that wins the channels in the odd time slot will continue to use the channels in the following even time slot.  A new round of channel competition will start in the next odd time slot following these two time slots. 

Since a licensed channel is accessed with probability $P_1$ in the odd time slot, we modify the 
tolerance threshold $\gamma_m$ as $\gamma'_m=\gamma_m/P_1$, such that the maximum allowable collision requirement can still be satisfied. 
In the even time slot, the channels will continue to be used by the winning cooperative relay link, i.e., to be accessed with probability 1. Therefore, the tolerance threshold is still $\gamma_m$ for the even time slots.

\subsubsection{Capacity Analysis \label{subsec:cap}}

Once the CR transmitter wins the competition, as indicated by a received CTS, it begins to send data over the licensed channels that are inferred to be idle (i.e., $D_m = 0$) in the transmission phase.  We assume the {\em channel bonding and aggregation} technique is used, such that multiple channels can be used collectively by a CR node for data transmission~\cite{Corderio06, Su08}. 

With DF, the winning CR transmitter uses all the available channels to transmit to the relay in the odd time slot. In the following even time slot, the CR transmitter stops transmission, while the relay uses the available channels in the even time slot to forward data to the receiver. If the number of available channels in the even time slot is equal to or greater than that in the odd time slot, the relay uses the same number of channels to forward all the received data. Otherwise, the relay uses all the available channels to forward part of the received data; the 
excess data will be dropped due to limited channel resource in the even time slot. The dropped data will be retransmitted in some future odd time slot by the transmitter. 

With AF, no matter it is an odd or even time slot, the CR transmitter always uses half of the available licensed channels to transmit to the relay.  The relay uses one of its transceivers to receive from the chosen half of the available channels. Simultaneously, it uses the other transceiver to forward the received data to the receiver using the remaining half of the available channels. 

Let $D_m^{od}$ and $D_m^{ev}$ be the decision variables of channel $m$ in the odd and even time slot, respectively (see (\ref{eq:DefDeci})).  Let $S_m^{od}$ and $S_m^{ev}$ be the status of channel $m$ in the odd and even time slot, respectively.  We have,
\begin{eqnarray}
&& \hspace{-0.1in} \Pr\{D_m^{od}=i,S_m^{od}=j,D_m^{ev}=k,S_m^{ev}=l\} \label{eq:SDODEV} \\
&&\hspace{-0.25in} = \Pr\{D_m^{ev}=k|S_m^{ev}=l\} \Pr\{D_m^{od}=i|S_m^{od}=j\} \times \nonumber \\ 
&& \hspace{-0.1in} \Pr\{S_m^{ev}=l|S_m^{od}=j\} \Pr\{S_m^{od}=j\}, \mbox{ for } i,j,k,l\in \{0,1\}. \nonumber
\end{eqnarray}      
where $\Pr\{S_m^{od}=j\}$ are the probabilities that channel $m$ is busy or idle, $\Pr\{S_m^{ev}=l\:|\:S_m^{od}=j\}$ are the channel $m$ transition probabilities.  $\Pr\{D_m^{ev}=k\:|\:S_m^{ev}=l\}$ and $\Pr\{D_m^{od}=i\:|\:S_m^{od}=j\}$ can be computed as in (\ref{eq:PrD0H1}) and (\ref{eq:PrD0H0}).

Let $N_{DF}$, $N_{AF}$ and $N_{DL}$ be the number of frames successfully delivered to the receiver in the two consecutive time slots using DF, AF and DL, respectively. Define  $\bar{S}_m^{od}=1-S_m^{od}$, $\bar{S}_m^{ev}=1-S_m^{ev}$, $\bar{D}_m^{od}=1-D_m^{od}$ and $\bar{D}_m^{ev}=1-D_m^{ev}$. We have
\begin{eqnarray} 
N_{DF} &=& \left( \mbox{$\sum_{m=1}^M$} \bar{S}_m^{od}\bar{D}_m^{od} \right) \wedge \left( \mbox{$\sum_{m=1}^M$} \bar{S}_m^{ev}\bar{D}_m^{ev} \right) \label{eq:DF} \\
N_{AF} &=& \left\lfloor \frac{1}{2}\mbox{$\sum_{m=1}^M$} \bar{S}_m^{od}\bar{D}_m^{od}\right\rfloor + \left\lfloor\frac{1}{2}\mbox{$\sum_{m=1}^M$} \bar{S}_m^{ev}\bar{D}_m^{ev}\right\rfloor \label{eq:AF} \\
N_{DL} &=& \left( \mbox{$\sum_{m=1}^M$} \bar{S}_m^{od}\bar{D}_m^{od} \right) + \left( \mbox{$\sum_{m=1}^M$} \bar{S}_m^{ev}\bar{D}_m^{ev} \right), \label{eq:DL}
\end{eqnarray}
where $x\wedge y$ represents the minimum of $x$ and $y$, and $\left\lfloor x \right\rfloor$ means the maximum integer that is not larger than $x$.

As discussed, the probability that a frame can be successfully delivered 
is $P_{DF}^k$, $P_{AF}^k$, or $P_{DL}^k$ for the three schemes, respectively. 
Recall that spectrum resources are allocated distributedly for every pair of two consecutive time slots.  We derive the capacity for the three cooperative relay strategies as
%
\begin{eqnarray} \label{eq:Cdfafdl}
  C_{DF} &=& \mbox{E}\left[ N_{DF} \right] \cdot \mbox{$\sum_{k=1}^N$} (P_{DF}^kP_1L)/(2NT_s) \\ 
  C_{AF} &=& \mbox{E}\left[ N_{AF} \right] \cdot \mbox{$\sum_{k=1}^N$} (P_{AF}^kP_1L)/(2NT_s) \\
  C_{DL} &=& \mbox{E}\left[ N_{DL} \right] \cdot \mbox{$\sum_{k=1}^N$} (P_{DL}^kP_1L)/(2NT_s),   
\end{eqnarray}
where $L$ is the packet length and $T_s$ is the duration of a time slot.  The expectations are computed using the results derived in (\ref{eq:SDODEV}) $\sim$ 
(\ref{eq:DL}).

\subsection{Performance Evaluation \label{sec:sim_relay}}

We evaluate the performance of the cooperative relay strategies with analysis and simulations.  The analytical capacities of the schemes are obtained with the analysis presented in Section~\ref{sec:ana_relay}.  The actual throughput is obtained using MATLAB simulations.  The simulation parameters and their values are listed in Table~\ref{tb:SimPara}, unless specified otherwise. We consider $M=5$ licensed channels and a CR network with seven cooperative relay links.  The channels have identical parameters for the Markov chain models. 
Each point in the simulation curves is the average of $10$ simulation runs with different random seeds.  We plot $95\%$ confidence intervals for the simulation results, 
which are negligible in all the cases. 

\begin{table}[!t]
\begin{center}
\caption{Simulation Parameters and Values}
\begin{tabular}{lll}
\hline
{\em Symbol} & {\em Value} & {\em Definition} \\
\hline
$M$ & 5 & number of licensed channels \\
$\lambda$ & 0.7 & channel transition probability \\
       & & from idle to idle \\
$\mu$ & 0.2 & channel transition probability \\
       & & from busy to idle \\
$\eta$ & 0.6 & channel utilization \\
$\gamma$ & 0.08 & maximum allowable collision \\
       & & probability \\
$N$ & 7 & number of CR cooperative relay \\
       & & links \\
$P_s$ & $10$ dBm & transmit power of the CR \\
       & & transmitters \\
$P_r$ & $10$ dBm & transmit power of CR relays \\
$L$ & $1$ kb & packet length \\
$T_s$ & $1$ ms & duration of a time slot\\
\hline
\end{tabular}
\label{tb:SimPara}
\end{center}
\end{table}

We first examine the impact of the number of licensed channels. 
To illustrate the effect of spectrum sensing, we let the decoding rate $P_{AF}^k$ be equal to $P_{DF}^k$. In Fig.~\ref{fig:ChanNum}, we plot the throughput of AF, DF, and DL under increased number of licensed channels. 
The analytical curves are upper bounds for the simulation curves in all the cases, and the gap between the two is reasonably small.  Furthermore, as the number of license channels is increased, the throughput of both AF and DF are increased.  The slope of the AF curves is larger than that of the DF curves.  There is a cross point between five and six, as predicted by both simulation and analysis curves.  This indicates that AF outperforms DF when the number of channels is large.  This is because AF is more flexible than DF in exploiting the idle channels in the two consecutive time slots.  
The DL analysis and simulation curves also increases with the number of channels, but with the lowest slope and the lowest throughput values.

\begin{figure}[!t]
\centering
\includegraphics[width=4.5in, height=3.0in]{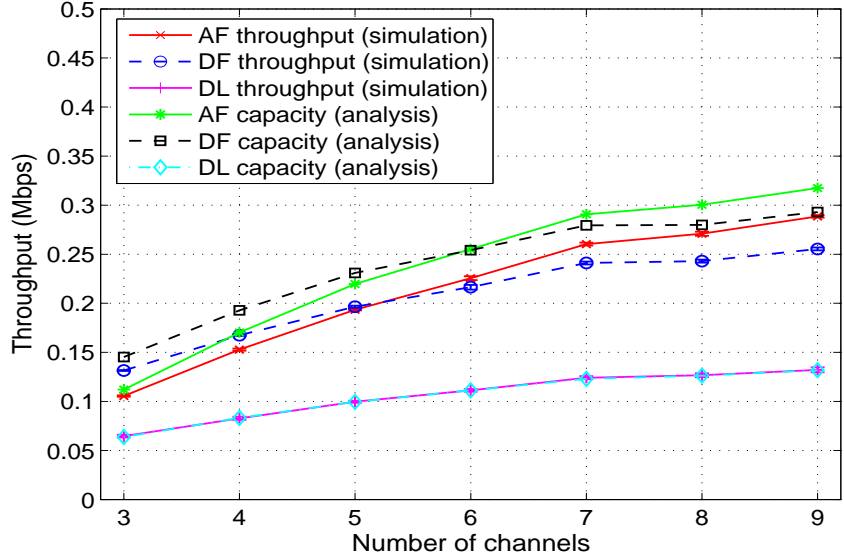}
\caption{Throughput performance versus number of licensed channels.}
\label{fig:ChanNum}
\end{figure}

In Fig.~\ref{fig:ChanUti}, we demonstrate the impact of channel utilization on the throughput of the schemes.  The channel utilization $\eta$ is increased from $0.3$ to $0.9$, when primary users get more active.  As $\eta$ is increased, the transmission opportunities for CR nodes are reduced and all the 
throughputs are degraded.  We find the throughputs of AF and DF are close to each other when the channel utilization is high.  AF outperforms DF in the low channel utilization region, but is inferior to DF in the high channel utilization region.  There is a cross point between the AF and DF curves between $\eta=0.5$ and $\eta=0.6$.  
When the channel utilization is low, there is a big gap between the cooperative relay curves and and the DL curves. 

\begin{figure}[!t]
\centering
\includegraphics[width=4.5in, height=3.0in]{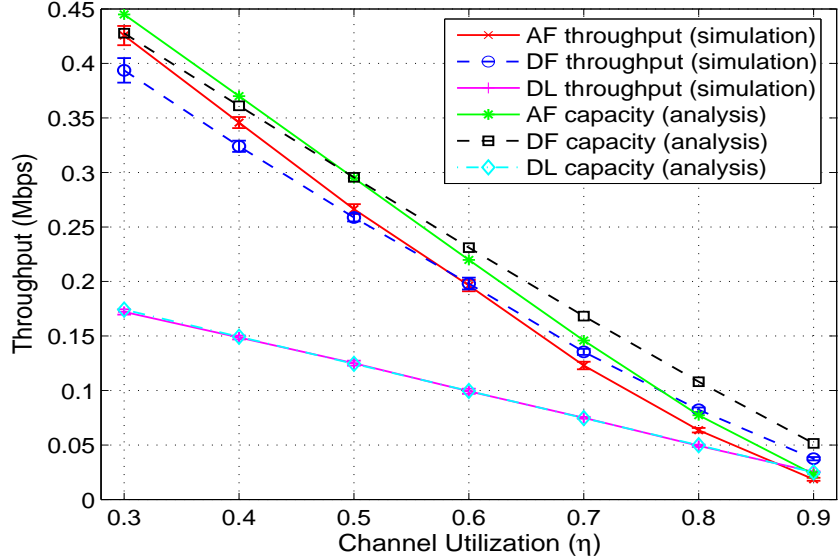}
\caption{Throughput performance versus primary user channel utilization.}
\label{fig:ChanUti}
\end{figure}

In Fig.~\ref{fig:ChanPower}, we examine the channel fading factor. We consider Rayleigh block fading channels, where the received power is exponentially distributed with a distance-dependent mean. 
We fix the transmitter power at 10 dBm, and increase the relay power from one dBm to 18 dBm. As the relay power is increased, the throughput is also increased since the SNR at the receiver is improved. We can see the increasing speed of AF is larger than that of DF, indicating that AF has superior performance than DF when the relay transmit power is large. The capacity analysis also demonstrate the same trend.  The throughput of DL does not depend on the relay node. Its throughput is better than that of AF and DF when the relay transmit power is low, since both AF and DF are limited by the relay-to-receiver link in this low power region.  However, the throughputs of AF and DF quickly exceed that of DL and grow fast as the relay-to-receiver link is improved with the increased relay transmit power. The considerable gaps between the cooperative relay link curves and the DL curves in Figs.~\ref{fig:ChanNum},~\ref{fig:ChanUti} and~\ref{fig:ChanPower} exemplify the diversity gain achieved by cooperative relays in CR networks. 

\begin{figure}[!t]
\centering
\includegraphics[width=4.5in, height=3.0in]{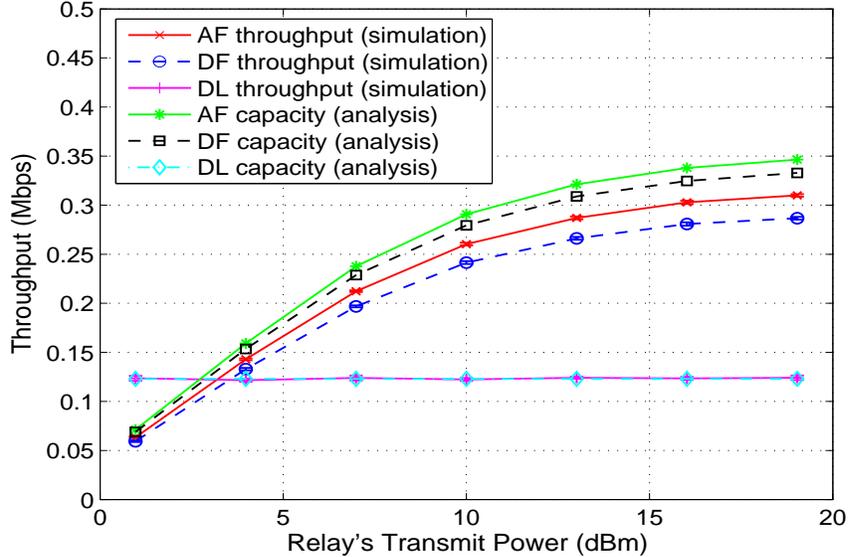}
\caption{Throughput performance versus transmit power of relay nodes.}
\label{fig:ChanPower}
\end{figure}

\section{Cooperative CR Networks with Interference Alignment}\label{sec:cr_video_coop}

In this section,  we investigate cooperative relay in CR networks using video as a reference application. We consider a base station (BS) and multiple relay nodes (RN) that collaboratively stream multiple videos to CR users within the network. It has been shown that the performance of a cooperative relay link is mainly limited by two factors: 
\begin{itemize}
	\item the {\em half-duplex operation}, since the BS--RN and the RN--user transmissions cannot be scheduled simultaneously on the same channel~\cite{Sendonaris03a}; and 
	\item the {\em bottleneck channel}, which is usually the BS--user and/or the RN--user channel, usually with poor quality due to obstacles, attenuation, multipath propagation and mobility~\cite{Laneman04}. 
\end{itemize}

To support high quality video service in such a challenging environment, we assume a well planned relay network where the RNs are connected to the BS with high-speed wireline links, and explore interference alignment to overcome the bottleneck channel problem~\cite{Hu12IC}. Therefore the video packets will be available at both the BS and the RNs before their scheduled transmission time, thus allowing advanced cooperative transmission techniques (e.g. interference alignment) to be adopted for streaming videos. 
In particular, we incorporate interference alignment to allow transmitters collaboratively send encoded signals to all CR users, such that undesired signals will be canceled and the desired signal can be decoded at each CR user. 

We present a stochastic programming formulation, as well as a reformulation that greatly reduces computational complexity. In the cases of a single licensed channel and multiple licensed channels with channel bonding, we develop an optimal distributed algorithm with proven convergence and convergence speed. In the case of multiple channels without channel bonding, we develop a greedy algorithm with a proven performance bound.

\subsection{Network Model and Assumptions}
The cooperative CR network is illustrated in Fig.~\ref{fig:netmod}. There is a CR BS (indexed $1$) and $(K-1)$ CR RNs (indexed from $2$ to $K$) deployed in the area to serve $N$ active CR users. Let $\mathcal{U} = \{1, 2, \cdots, N\}$ denote the set of active CR users. We assume that the BS and all the RNs are equipped with multiple transceivers: one is tuned to the common control channel and the others are used to sense multiple licensed channels at the beginning of each time slot, and to transmit encoded signals to CR users. 
We consider the case where each CR user has one software defined radio (SDR) based transceiver, 
which can be tuned to operate on any of the $(M+1)$ channels. If the channel bonding/aggregation techniques are used~\cite{Mahmoud09, Corderio06}, a transmitter can collectively use all the available channels and a CR user can receive from all the available channels simultaneously. Otherwise, only one licensed channel will be used by a transmitter and a CR user can only receive from a single chosen channel at a time.  

\begin{figure} [!t]
 \centering
 \includegraphics [width=4.5in]{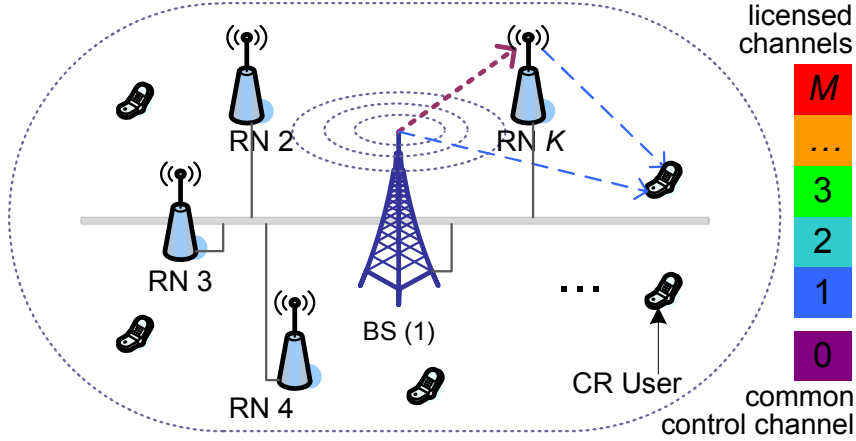}
 \caption{Illustration of the cooperative CR network.}
 \label{fig:netmod}
\end{figure}

Consider the three channels in a traditional cooperative relay link. Usually the BS and RNs are mounted on high towers, and the BS--RN channel has good quality due to line-of-sight (LOS) communications and absence of mobility.  On the other hand, a CR user is typically on the ground level. The BS--user and RN--user channels usually have much poorer quality due to obstacles, attenuation, multipath propagation and mobility. To support high quality video service, we assume a well planned relay network, where the RNs are connected to the BS via broadband wireline connections (e.g., as in femtocell networks~\cite{Hu12JSAC}). Alternatively, free space optical links can also be used to provide multi-gigabit rates between the BS and the RNs~\cite{Son10}. As a result, the video packets will always be available for transmission (with suitable channel coding and retransmission) at the RNs at their scheduled transmission time. To cope with the much poorer BS--user and RN--user channels, the BS and RNs adopt interference alignment to cooperatively transmit video packets to CR users, while exploiting the spectrum opportunities in the licensed channels.  

\subsubsection{Spectrum Access}
The BS and the RNs sense the licensed channels and exchange their sensing results over the common control channel during the sensing phase. Given $L$ sensing results obtained for channel $m$, the corresponding sensing result vector is $\vec{\Theta}_L^m = [\Theta_1^m, \Theta_2^m, \cdots, \Theta_L^m]$. Let $P_m^A(\vec{\Theta}_L^m):=P_m^A(\Theta_1^m, \Theta_2^m, \cdots, \Theta_L^m)$ be the conditional probability that channel $m$ is available, which can be computed iteratively as shown in our prior work~\cite{Hu12JSAC}:
\begin{eqnarray}\label{eq:Iteration1}
P_m^A(\Theta_1^m)&=&\left[1+\frac{\eta_m}{1-\eta_m}\times\frac{(\delta_1^m)^{1-\Theta_1^m}(1-\delta_1^m)^{\Theta_1^m}}{(\epsilon_1^m)^{\Theta_1^m}(1-\epsilon_1^m)^{1-\Theta_1^m}}\right]^{-1} \nonumber \\
P_m^A(\vec{\Theta}_l^m)&:=&P_m^A(\Theta_1^m,\Theta_2^m,\cdots,\Theta_l^m) \nonumber \\
&=& \left\{1+\left[\frac{1}{P_m^A(\Theta_1^m,\Theta_2^m,\cdots,\Theta_{l-1}^m)}-1\right]\times 
\frac{(\delta_l^m)^{1-\Theta_l^m}(1-\delta_l^m)^{\Theta_l^m}}{(\epsilon_l^m)^{\Theta_l^m}(1-\epsilon_l^m)^{1-\Theta_l^m}}\right\}^{-1},l\ge2. \nonumber
\end{eqnarray}

For each channel $m$, define an index variable $D_m(t)$ for the BS or RNs to access the channel in time slot $t$. That is, 
\begin{eqnarray}\label{eq:DefDm1}
  D_m(t)=\left\{\begin{array}{l l}
      $0$, & \mbox{access channel $m$ in time slot $t$} \\
      $1$, & \mbox{otherwise,}
                \end{array}\right. 
      m = 1, 2, \cdots, M. 
\end{eqnarray}

With sensing result $P_m^A(\vec{\Theta}_L^m)$, each channel $m$ will be opportunistically accessed. Let the probability be $P_m^D(\vec{\Theta}_L^m)$ that channel $m$ will be accessed in time slot $t$ (i.e., when $D_m(t)=0$). The optimal channel access probability can be computed as: 
\begin{eqnarray}\label{eq:PrDm1}
  P^D_m(\vec{\Theta}_L^m) = \min \left\{ \gamma_m / \left[ 1 - 
             P^A_m(\vec{\Theta}_L^m) \right], 1 \right\}.
\end{eqnarray}
Let $\mathcal{A}(t)$ be the set of available channels in time slot $t$. It follows that
$\mathcal{A}(t):=\{m\;|\;D_m(t)=0\}$.

\subsubsection{Interference Alignment \label{subsec:iac}}

We next briefly describe the main idea of interference alignment considered in this paper. Interested readers are referred to~\cite{Gollakota09, Li10} for insightful examples, a classification of various interference alignment scenarios, and practical considerations. 
 
Consider two transmitters (denoted as $s_1$ and $s_2$ ) and two receivers (denoted as $d_1$ and $d_2$). Let $X_1$ and $X_2$ be the signals corresponding to the packets to be sent to $d_1$ and $d_2$, respectively. With interference alignment, the transmitters $s_1$ and $s_2$ send compound signals $a_{1,1}X_1+a_{1,2}X_2$ and $a_{2,1}X_1+a_{2,2}X_2$, respectively, to the two receivers $d_1$ and $d_2$ simultaneously. If channel noise is ignored, the received signals $Y_1$ and $Y_2$ can be written as:
\begin{eqnarray}\label{eq:InterAlign}
\left[\begin{array}{c}
	Y_1\\
	Y_2
\end{array}\right]&=&
\left[\begin{array}{cc}
	G_{1,1} & G_{1,2}\\
	G_{2,1} & G_{2,2}
\end{array}\right]^{\mathsf{T}}
\left[\begin{array}{cc}
	a_{1,1} & a_{1,2}\\
	a_{2,1} & a_{2,2}
\end{array}\right]
\left[\begin{array}{c}
	X_1\\
	X_2
\end{array}\right] 
:= \mathbf{G}^{\mathsf{T}} \times \mathbf{A} \times \vec{X},
\end{eqnarray} 
where $G_{i,j}$ is the channel gain from transmitter $s_i$ to receiver $d_j$. 

From (\ref{eq:InterAlign}), it can be seen that both receivers can perfectly decode their signals if the transformation matrix $\mathbf{A}$ is chosen to be $\left\{\mathbf{G}^{\mathsf{T}}\right\}^{-1}$, i.e., the inverse of the channel gain matrix. With this technique, the transmitters are able to send packets simultaneously and the interference between the two concurrent transmissions can be effectively canceled at both receivers~\cite{Li10}.

\subsection{Problem Formulation \label{sec:form}}

We formulate the problem of interference alignment for scalable video streaming over cooperative CR networks in this section. As discussed in Section~\ref{sec:mod}, the video packets are available at both the BS and all the RNs before their scheduled transmission time; the BS and RNs adopt interference alignment to overcome the poor BS--CR user and RN--CR user channels. 

Let $X_j$ denote the signal to be transmitted to user $j$, which has unit power. As illustrated in Section~\ref{subsec:iac}, 
transmitter $k$ sends a compound signal $\sum_{j\in \mathcal{U}}a_{k,j}X_j$ to all active CR users, where $a_{k,j}$'s are the weights to be determined. Ignoring channel noise, we can compute the received signal $Y_n$ at a user $n$ as:  
\begin{eqnarray}\label{eq:RecvSig}
  Y_n&=&\sum_{k=1}^K G_{k,n} \sum_{j=1}^N a_{k,j}X_j  
  =\sum_{k=1}^K \sum_{j=1}^N  a_{k,j} G_{k,n}X_j  
  \nonumber \\
  &=&
  \sum_{j=1}^N X_j\sum_{k=1}^K a_{k,j} G_{k,n}, \; n=1,2,\cdots,N,
\end{eqnarray}
where $G_{k,n}$ is the channel gain from the BS (i.e., $k=1$) or an RN $k$ to user $n$.  For user $n$, only signal $X_n$ should be decoded and the coefficients of all other signals should be forced to zero. The {\em zero-forcing constraints} can be written as:
\begin{eqnarray}\label{eq:OrthSig}
  \sum_{k=1}^K a_{k,j} G_{k,n}=0,  \;\;\;\mbox{for all } j \neq n. 
\end{eqnarray}

Usually the total transmit power of the BS and every RN is limited by a peak power $P_{max}$. Since $X_j$ has unit power, for all $j$, the power of each transmitted signal is the square sum of all the coefficients $a_{k,j}^2$. The {\em peak power constraint} can be written as
\begin{eqnarray}\label{eq:MaxPow}
  \sum_{j=1}^N |a_{k,j}|^2 \le P_{max}, \;\;\; k=1, \cdots, K.
\end{eqnarray}

Recall that each CR user has one SDR transceiver that can be tuned to receive from any of the $(M+1)$ channels, when channel bonding is not adopted. Let $b_j^m$ be a binary variable indicating that user $j$ selects licensed channel $m$. It is defined as
\begin{eqnarray}\label{eq:Defbnm}
  b_j^m=\left\{\begin{array}{l l}
        $1$, & \mbox{if user $n$ receives from channel $m$}\\
        $0$, & \mbox{otherwise,} 
               \end{array}\right. 
        j=1,\cdots,N, \; m=1,\cdots,M.  
\end{eqnarray}
Then, we have the following {\em transceiver constraint}:
\begin{eqnarray}\label{eq:OneChan2}
  \sum_{m\in \mathcal{A}(t)} b_j^m \leq 1, \;\;\; j=1,\cdots,N.  
\end{eqnarray} 
After introducing the channel selection variables $b_j^m$'s, the overall channel gain becomes
\begin{eqnarray}\label{eq:OrthSig2}
  G_{k,j}= \sum_{m\in \mathcal{A}(t)}b_j^m H_{k,j}^m,  
\end{eqnarray}
where $H_{k,j}^m$ is the channel gain from the BS (i.e., $k=1$) or an RN $k$ to user $j$ on channel $m$.

Let $w_j^t$ be the PSNR of user $j$'s reconstructed video at the beginning of time slot $t$ and $W_j^t$ the PSNR of user $j$'s reconstructed video at the end of time slot $t$. In time slot $t$, $w_j^t$ is already known, while $W_j^t$ is a random variable depending on the resource allocation and primary user activity during the time slot. That is, $w_j^{t+1}$ is a realization of $W_j^t$. 

The quality of reconstructed MGS video can be modeled with a linear equation~\cite{Wien07}:
\begin{equation}\label{eq:RateDist}
  W(R)=\alpha + \beta \times R,
\end{equation}
where $W(R)$ is the average peak signal-to-noise ratio (PSNR) of the reconstructed MGS video, $R$ is the average data rate, and $\alpha$ and $\beta$ are constants depending on the specific video sequence and codec. 

%

We formulate a multistage stochastic programming problem to maximize the sum of expected logarithm of the PSNR's at the end of the GOP, i.e., $\sum_{j=1}^N \mathbb{E}\left[\log(W_j^T)\right]$, for proportional fairness among the video sessions~\cite{Kelly98}. 
It can be shown that the multistage stochastic programming problem can be decomposed into $T$ serial sub-problems, one for each time slot $t$, as~\cite{Hu10TW}: 
\begin{eqnarray} \label{eq:ObjFun1}
\mbox{maximize:  } && \hspace{-0.2in} \sum_{j=1}^N \mathbb{E} \left[\log(W_j^t)|w_j^t \right] \\
\mbox{subject to:} && \hspace{-0.2in} W_j^t=w_j^t+\Psi_j^t \\
                   && \hspace{-0.2in} b_j^m \in \{0,1\}, \; a_{k,j} \geq 0, \; \mbox{for all } m, j, k  
                      \label{eq:ba} \\
                   && \hspace{-0.2in} \mbox{Constraints~(\ref{eq:OrthSig}),~(\ref{eq:MaxPow}) and~(\ref{eq:OneChan2})}, \nonumber
\end{eqnarray} 
where $\Psi_j^t$ is a random variable that depends on spectrum sensing, power allocation, and channel selection in time slot $t$. This is a mixed integer nonlinear programming problem (MINLP), with binary variables $b_j^m$'s and continuous real variables $a_{k,j}$'s. 

In particular, $\Psi_j^t$ can have two possible values: (i) zero, if the packet is not successfully received due to collision with primary users; (ii) the PSNR increase achieved in time slot $t$ if the packet is successfully received, denoted as $\lambda_j^t$. The PSNR increase can be computed as:
\begin{eqnarray}\label{eq:lambda}
  \lambda_j^t = \frac{\beta_j B}{T} \log_2 \left( 1 + \frac{1}{N_0} \left( \sum_{k=1}^K 
                a_{k,j}G_{k,j} \right)^2 \right),
\end{eqnarray}
where $N_0$ is the noise power and $B$ is the channel bandwidth.  

User $j$ can successfully receive a video packet from channel $m$ if it tunes to channel $m$ (i.e., $b_j^m=1$) and the BS and RNs transmit on channel $m$ (i.e., with probability $P^D_m(\vec{\Theta}_L^m)$). The probability that user $j$ successfully receives a video packet, denoted as $P_j^t$, is
\begin{eqnarray}\label{eq:Pjt}
  P_j^t = \sum_{m\in \mathcal{A}(t)}b_j^m P^D_m(\vec{\Theta}_L^m).
\end{eqnarray} 
Therefore, we can expand the expectation in (\ref{eq:ObjFun1}) to obtain a reformulated problem:
\begin{eqnarray} \label{eq:ExpTerm}
\hspace{-0.2in} \mbox{maximize:  } && \hspace{-0.2in} \sum_{j=1}^N \mathbb{E} \left[ P_j^t \log(w_j^t+\lambda_j^t)+(1-P_j^t)\log(w_j^t) \right] \\
\hspace{-0.3in} \mbox{subject to:} && \hspace{-0.2in}  \mbox{constraints~(\ref{eq:OrthSig}),~(\ref{eq:MaxPow}),~(\ref{eq:OneChan2}), 
     and~(\ref{eq:ba})}. \nonumber
\end{eqnarray}

\subsection{Solution Algorithms \label{sec:sol}}

In this section, we develop effective solution algorithms to the stochastic programming problem~(\ref{eq:ObjFun1}). In Section~\ref{subsec:onec}, we first consider the case of a single licensed channel, and derive a distributed, optimal algorithm with guaranteed convergence and bounded convergence speed. We then address the case of multiple licensed channels. If channel bonding/aggregation techniques are used~\cite{Mahmoud09, Corderio06}, the distributed algorithm in Section~\ref{subsec:onec} can still be applied to achieve optimal solutions. We finally consider the case of multiple licensed channels without channel bonding, and develop a greedy algorithm with a performance lower bound in Section~\ref{subsec:mcnocbd}.

\subsubsection{Case of a Single Channel \label{subsec:onec}}

\paragraph{Property \label{subsubsec:propty}}
Consider the case when there is only one licensed channel, i.e., when $M=1$. The $K$ transmitters, including the BS and $(K-1)$ RNs, send video packets to active users using the licensed channel when it is sensed idle. 

\begin{definition} A set of vectors is {\em linearly independent} if none of them can be written as a linear combination of the other vectors in the set~\cite{Strang09}.
\end{definition}

For user $j$, the weight and channel gain vectors are: $\vec{a}_j=[a_{1,j},a_{2,j},\cdots,a_{K,j}]^\mathsf{T}$ and $\vec{G}_j=[G_{1,j},G_{2,j},\cdots,G_{K,j}]^\mathsf{T}$, where $\mathsf{T}$ denotes {\em matrix transpose}. 
Due to spatial diversity, we assume that the $\vec{G}_j$ vectors are linearly independent~\cite{Tse05}. 

\begin{lemma} \label{lemma1:1}
To successfully decode each signal $X_j$, $j=1, 2, \cdots, N$, the number of active users $N$ should be smaller than or equal to the number of transmitters $K$.
\end{lemma}
%
\begin{proof}
From (\ref{eq:OrthSig}), it can be seen that $\vec{a}_j$ 
is orthogonal to the $(N-1)$ vectors $\vec{G}_n$'s, for $n \neq j$.
Since $\vec{a}_j$ is a $K$ by $1$ vector, there are at most $(K-1)$ vectors that are orthogonal to $\vec{a}_j$. Since the $\vec{G}_j$ vectors are linearly independent, it follows that $(N-1) \le (K-1)$ and therefore $N \le K$.  
\end{proof}

According to Lemma~\ref{lemma1:1}, the following additional constraints should be enforced for the channel selection variables. 
\begin{eqnarray}\label{eq:MaxSender}
  \sum_{j=1}^N b_j^m \leq K, \;\;\;\mbox{for all } m \in \mathcal{A}(t).  
\end{eqnarray} 
That is, the number of active users receiving from any channel $m$ cannot be more than the number of transmitters on that channel, which is $K$ in the single channel case and less than or equal to $K$ in the multiple channels case.  
We first assume that $N$ is not greater than $K$, and will remove this assumption in the following subsection.

\paragraph{Reformulation and Complexity Reduction \label{subsubsec:reform}}
With a single channel, all active users receive from channel 1. Therefore $b_j^1=1$, and $b_j^m=0$, for $m>1$, $j=1,2,\cdots,N$. The formulated problem is now reduced to a nonlinear programming problem with constraints (\ref{eq:OrthSig}), (\ref{eq:MaxPow}), and~(\ref{eq:ba}).  
If the number of active users is $N=1$, the solution is straightforward: all the transmitters send the same signal $X_1$ to the single user using their maximum transmit power $P_{max}$. 

In general, the reduced problem can be solved with the dual decomposition technique~\cite{Bertsekas99} (i.e., a primal dual algorithm).  This problem has $K\times N$ primal variables (i.e., the $a_{k,j}$'s), and we need to define $N(N-1)$ dual variables (or, Lagrangian Multipliers) for constraints (\ref{eq:OrthSig}) and $K$ dual variables for constraints (\ref{eq:MaxPow}). These numbers could be large for even moderate-sized systems. Before presenting the solution algorithm, we first derive a reformulation of the original problem~(\ref{eq:ExpTerm}) that can greatly reduce the number of primal and dual variables, such that the computational complexity can be reduced.

\begin{lemma} Each vector $\vec{a}_j = [a_{1,j}, a_{2,j}, \cdots, a_{K,j}]^\mathsf{T}$ can be represented by the linear combination of $r$ nonzero, linearly independent vectors, where $r=K-N+1$.
\end{lemma}

\begin{proof} 
From (\ref{eq:OrthSig}), each vector $\vec{a}_j$ is orthogonal to $\vec{G}_i$ where $j\neq i$. Define a reduced matrix $\mathbf{G}_{-j}$ obtained by deleting $\vec{G}_j$ from $\mathbf{G}$, i.e., $\mathbf{G}_{-j}=[\vec{G}_1,\cdots, \vec{G}_{j-1},\vec{G}_{j+1},\cdots,\vec{G}_N]$. Then $\vec{a}_j$ is a solution to the homogeneous linear system $\mathbf{G}_{-j}^\mathsf{T}\vec{x}=0$. Since we assume that the $\vec{G}_i$'s are all linearly independent, the columns of $\mathbf{G}_{-j}$ are also linearly independent~\cite{Strang09}. Thus the rank of $\mathbf{G}_{-j}$ is $(N-1)$. The solution belongs to the null space of $\mathbf{G}_{-j}$. The dimension of the null space is $r=K-(N-1)$ according to the Rank-nullity Theorem~\cite{Strang09}. Therefore, each $\vec{a}_j$ can be presented by the linear combination of $r$ linearly independent vectors.
\end{proof}

Let $\mathbf{e}_j = \{\vec{e}_{j,1}, \vec{e}_{j,2}, \cdots, \vec{e}_{j,r}\}$ be a {\em basis} for the null space of $\mathbf{G}_{-j}$. There are many methods to obtain the basis, such as Gaussian Elimination. However, we show that it is not necessary to solve the homogeneous linear system $\mathbf{G}_{-j}^\mathsf{T}\vec{x}=0$  to get the basis for every different $j$ value. Therefore the computational complexity can be further reduced. 

Our algorithm for computing a basis is shown in Table~\ref{tab:GenBasis}. In Steps 1--6, we first solve the homogeneous linear system $\mathbf{G}^\mathsf{T}\vec{x}=0$ to get a basis $[\vec{v}_1,\vec{v}_2,\cdots,\vec{v}_{K-N}]$. Note that if $K$ is equal to $N$, the basis is the empty set $\emptyset$. We then set the $K-N$ basis vectors to be the first $K-N$ vectors in all the basses $\mathbf{e}_j$, $j=1,2,\cdots,N$. In Step 8, we orthogonalize each $\mathbf{G}_{-j}$ and obtain $(N-1)$ orthogonal vectors $\vec{\omega}_{j,i}$, $i=1,2,\cdots,N-1$. Finally in Step 9, we let the 
$r$th vector $\vec{e}_{j,r}$ be orthogonal to all the $\vec{\omega}_{j,i}$'s by subtracting all the projections on each $\vec{\omega}_{j,i}$ from $\vec{G}_j$ (recall that $r=K-N+1$). The 
operation is:
\begin{eqnarray}\label{eq:LastEvec}
\vec{e}_{j,r} = \vec{e}_{j,N-K+1} = \vec{G}_j - \sum_{i=1}^{N-1} \frac{\vec{G}_j^\mathsf{T} \vec{\omega}_{j,i}} {\vec{\omega}_{j,i}^\mathsf{T} \vec{\omega}_{j,i}} \vec{\omega}_{j,i}.
\end{eqnarray}

\begin{table} [!t]
\begin{center}
\caption{Basis Computation Algorithm}
\begin{tabular}{ll}
\hline
1: & IF ($K > N$) \\
2: & $\;\;\;$ Solve homogeneous linear system $\mathbf{G}^\mathsf{T}\vec{x}=0$ and get \\
   & $\;\;\;$ basis $[\vec{v}_1,\cdots,\vec{v}_{K-N}]$; \\
3: & $\;\;\;$ FOR $i=1$ to $K-N$ \\
4: &$ \;\;\;\;\;\;$ $\vec{e}_{j,i} = \vec{v}_i$, for all $j$; \\
5: & $\;\;\;$ END FOR\\
6: & END IF\\
7: & FOR $j=1$ to $N$ \\
8: & $\;\;\;$ Orthogonalize $\mathbf{G}_{-j}$ 
              and get $(N-1)$ orthogonal vectors $\vec{w}_{j,i}$'s; \\
9: & $\;\;\;$ Calculate $\vec{e}_{j,r}$ as in (\ref{eq:LastEvec}); \\
10: & END FOR\\
\hline
\end{tabular}
\label{tab:GenBasis}
\end{center} 
\end{table}

\begin{lemma} The solution space constructed by the basis $[\vec{v}_1,\vec{v}_2,\cdots,\vec{v}_{K-N}]$ is a sub-space of the solution space of $\mathbf{G}_{-j}^\mathsf{T}\vec{x}=0$ for all $j$.
\end{lemma}
\begin{proof} 
It is easy to see that each vector $\vec{v}_i$ is a solution of $\mathbf{G}_{-j}^\mathsf{T}\vec{x}=0$ by substituting $\vec{x}$ with $\vec{v}_i$, for $i=1,2,\cdots,K-N$. 
\end{proof}

\begin{lemma} The vectors $[\vec{v}_1, \vec{v}_2, \cdots, \vec{v}_{K-N}, \vec{e}_{j,r}]$ computed in Table~\ref{tab:GenBasis} is a basis of the null space of $\mathbf{G}_{-j}$.
\end{lemma}
\begin{proof} 
Obviously, the $\vec{v}_i$'s are linearly independent. From (\ref{eq:LastEvec}), it is easy to verify that $\vec{e}_{j,r}$ is orthogonal to all the $\vec{\omega}_{j,i}$'s. Therefore, $\vec{e}_{j,r}$ is also a solution to system  $\mathbf{G}_{-j}^\mathsf{T}\vec{x}=0$. Since $\vec{G}_j$ and $\vec{\omega}_{j,i}$ are orthogonal to all the $\vec{v}_i$'s, and $\vec{e}_{j,r}$ is a linear combination of $\vec{G}_j$ and $\vec{\omega}_{j,i}$, $\vec{e}_{j,r}$ is also orthogonal and linearly independent to all the $\vec{v}_i$'s. 
The conclusion follows. 
\end{proof}

Define coefficients $\vec{c}_j = [c_{j,1}, c_{j,2}, \cdots, c_{j,r}]^\mathsf{T}$. Then we can represent $\vec{a}_j$ as a linear combination of the basis vectors, i.e., $\vec{a}_j=\sum_{l=1}^r c_{j,l}\vec{e}_{j,l}=\mathbf{e}_j \vec{c}_j$. Eq.~(\ref{eq:lambda}) can be rewritten as
\begin{eqnarray}\label{eq:lambda2}
  \lambda_j^t=\frac{\beta_j B}{T}\log_2 \left( 1 + \frac{1}{N_0} \left( \vec{c}_j^\mathsf{T}   
    \mathbf{e}_j^\mathsf{T} \vec{G}_j \right)^2 \right) 
             =\frac{\beta_j B}{T}\log_2 \left( 1 + \frac{1}{N_0}
              \left( c_{j,r}\vec{e}_{j,r}^\mathsf{T} \vec{G}_j \right)^2 \right).
\end{eqnarray}
The second equality is because the first $K-N$ column vectors in $\mathbf{e}_j$ are orthogonal to $G_j$. The random variable $W_j^t$ in the objective function 
now only depends on $c_{j,r}$. 
The peak power constraint can be revised as:
\begin{eqnarray}\label{eq:MaxPow2}
  \sum_{j=1}^N [\mathbf{e}_j(k) \vec{c}_j]^2 \le P_{max}, 
               \;\;\; k=1,\cdots,K, 
\end{eqnarray}
where $\mathbf{e}_j(k)$ is the $k$th row of matrix $\mathbf{e}_j$. 

With such a reformulation, the number of primal and dual variables can be greatly reduced. In Table~\ref{tab:NumVar}, we show the numbers of variables 
in the original problem and in the reformulated problem. 
The number of primary variables is reduced from $KN$ to $(K-N+1) N$, and the number of dual variables is reduced from $N(N-1)+K$ to $K$. Such reductions result in greatly reduced computational complexity.

\begin{table} [!t]
\begin{center}
\caption{Comparison of Computational Complexity}
\begin{tabular}{l|l|l}
\hline
 & Original Problem & Reformulated Problem \\
\hline 
 Primal Variables & $KN$ &  $(K-N+1) N$ \\
 Dual Variables & $N(N-1)+K$ & $K$ \\
\hline
\end{tabular}
\label{tab:NumVar}
\end{center}
\end{table}

\paragraph{Distributed Algorithm \label{subsubsec:sol1}}

To solve the reformulated problem, we define non-negative dual variables  $\vec{\mu}=[\mu_1,\cdots,\mu_K]^\mathsf{T}$ for the inequality constraints. The Lagrangian function is
\begin{eqnarray}\label{eq:LagFun1}
\mathcal{L}(\mathbf{c},\vec{\mu})&=&\sum_{j=1}^N \mathbb{E}\left[\log(W_j^t(c_{j,r}))|w_j^t\right]+ 
\sum_{k=1}^K \mu_k (P_{max}-\sum_{j=1}^N [\mathbf{e}_j(k) \vec{c}_j]^2) \nonumber\\
&=& \sum_{j=1}^N \mathcal{L}_j(\vec{c}_j,\vec{\mu})+ P_{max}\sum_{k=1}^K \mu_k,
\end{eqnarray} 
where $\mathbf{c}$ is a matrix consisting of all column vector $\vec{c}_j$'s and 
$$
\mathcal{L}_j(\vec{c}_j,\vec{\mu})=\mathbb{E}\left[\log(W_j^t(c_{j,r}))|w_j^t\right]-\sum_{k=1}^K\mu_k[\mathbf{e}_j(k) \vec{c}_j]^2.
$$

The corresponding problem can be decomposed into $N$ sub-problems and solved iteratively \cite{Bertsekas99}. In Step $\tau\ge 1$, for given vector $\vec{\mu}(\tau)$, each CR user solves the following sub-problem using local information
\begin{eqnarray}\label{eq:PrimProbm}
\vec{c}_j(\tau)=\arg\max \mathcal{L}_j(\vec{c}_j,\vec{\mu}(\tau)).
\end{eqnarray} 
Obviously, the objective function in (\ref{eq:PrimProbm}) is concave. Therefore, there is a unique optimal solution. 
The CR users then exchange their solutions over the common control channel. To solve the primal problem, we adopt the gradient method~\cite{Bertsekas99}. 
\begin{eqnarray}\label{eq:GradPrim}
  \vec{c}_j(\tau+1)=\vec{c}_j(\tau)+\phi\nabla   
  \mathcal{L}_j(\vec{c}_j(\tau),\vec{\mu}(\tau)),
\end{eqnarray} 
where $\nabla \mathcal{L}_j(\vec{c}_j(\tau),\vec{\mu}(\tau))$ is the gradient of the primal problem and $\phi$ is a small positive step size.

The master dual problem for a given $\mathbf{c}(\tau)$ is:
\begin{eqnarray}\label{eq:DualProbm}
  \min_{\mu_i \ge 0,i=1,\cdots,K} q(\vec{\mu}) = \sum_{j=1}^N 
     \mathcal{L}_j(\vec{c}_j(\tau),\vec{\mu}) +
     P_{max}\sum_{k=1}^K \mu_k.
\end{eqnarray} 
Since the Lagrangian function is differentiable, the subgradient iteration method can be adopted.
\begin{eqnarray}\label{eq:GradDual}
\vec{\mu}(\tau+1)= [\vec{\mu}(\tau)-\rho(\tau) \vec{g}(\tau)]^+,
\end{eqnarray} 
where $\rho(\tau)=\frac{q(\vec{\mu}(\tau))-q(\vec{\mu}^\ast)}{||\vec{g}(\tau)||^2}$ is a positive step size, $\vec{\mu}^\ast$ is the optimal solution, $\vec{g}(\tau)=\nabla q(\vec{\mu}(\tau))$ is the gradient of the dual problem, and $[\cdot]^+$ denotes the projection onto the nonnegative axis. Since the optimal solution $\vec{\mu}^\ast$ is unknown a priori, we choose the mean of the objective values of the primal and dual problems as an estimate for $\vec{\mu}^\ast$ in the algorithm. The updated $\mu_k(\tau+1)$ will again be used to solve the sub-problems (\ref{eq:PrimProbm}). Since the problem is convex, we have strong duality; the duality gap between the primal and dual problems will be zero. The distributed algorithm is shown in Table~\ref{tab:SingleChan}, where $0 \le \kappa \ll 1$ is a threshold for convergence.

\begin{table} [!t]
\begin{center}
\caption{Algorithm for the Case of a Single Channel}
\begin{tabular}{ll}
\hline
1: & IF ($N=1$) \\
2: & $\;\;\;$ Set $a_{k,j}$ to $P_{max}$ for all $k$; \\
3: & ELSE \\
4: & $\;\;\;$ Set $\tau=0$; $\vec{\mu}(0)$ to positive values; $\mathbf{c}(0)$ to random values; \\
5: & $\;\;\;$ Compute bases $\mathbf{e}_j$'s as in Table~\ref{tab:GenBasis}; \\
6: & $\;\;\;$ DO \\
7: & $\;\;\;\;\;\;$ $\tau=\tau+1$ ; \\
8: & $\;\;\;\;\;\;$ Compute $\vec{c}_j(\tau)$ as in (\ref{eq:GradPrim}); \\
9: & $\;\;\;\;\;\;$ Broadcast $\vec{c}_j(\tau)$ on the common control channel; \\
10:& $\;\;\;\;\;\;$ Update $\vec{\mu}(\tau)$ as in (\ref{eq:GradDual}); \\
11:& $\;\;\;$ WHILE ($||\vec{\mu}(\tau)-\vec{\mu}(\tau-1)||> \kappa $); \\
12:& $\;\;\;$ Compute $a_{k,j}$'s; \\
13:& END IF\\
\hline
\end{tabular}
\label{tab:SingleChan}
\end{center}
\end{table}

\paragraph{Performance Analysis \label{subsubsec:ana1}}
We analyze the performance of the distributed algorithm in this section. In particular, we prove that it converges to the optimal solution at a speed faster than $\sqrt{1/\tau}$ as $\tau$ goes to infinity.  

\begin{theorem} \label{th:1} The series $q(\vec{\mu}(\tau))$ converges to $q(\vec{\mu}^\ast)$ as $\tau$ goes to infinity and the square error sum $\sum_{\tau=1}^{\infty}(q(\vec{\mu}(\tau))-q(\vec{\mu}^\ast))^2$ is bounded.
\end{theorem}
\begin{proof} 
For the optimality gap, we have:
\begin{eqnarray}\label{eq:ExpdIneq}
&&||\vec{\mu}(\tau+1)-\vec{\mu}^\ast||^2=||[\vec{\mu}(\tau)-\rho(\tau) \vec{g}(\tau)]^+-\vec{\mu}^\ast||^2 \nonumber \\
&\le&||\vec{\mu}(\tau)-\rho(\tau) \vec{g}(\tau)-\vec{\mu}^\ast||^2 \nonumber \\
&=&||\vec{\mu}(\tau)-\vec{\mu}^\ast||^2-2\rho(\tau)(\vec{\mu}(\tau)-\vec{\mu}^\ast)^\mathsf{T}\vec{g}(\tau)+
(\rho(\tau))^2||\vec{g}(\tau)||^2 \nonumber \\
&=&||\vec{\mu}(\tau)-\vec{\mu}^\ast||^2-2\rho(\tau)(q(\vec{\mu}(\tau))-q(\vec{\mu}^\ast))+ 
(\rho(\tau))^2||\vec{g}(\tau)||^2. \nonumber 
\end{eqnarray}
Since the step size is $\rho(\tau)=\frac{q(\vec{\mu}(\tau))-q(\vec{\mu}^\ast)}{||\vec{g}(\tau)||^2}$, it follows that
\begin{eqnarray} 
||\vec{\mu}(\tau+1)-\vec{\mu}^\ast||^2 &\le& ||\vec{\mu}(\tau)-\vec{\mu}^\ast||^2-\frac{(q(\vec{\mu}(\tau))-q(\vec{\mu}^\ast))^2}{||\vec{g}(\tau)||^2}\nonumber \\
&\leq& ||\vec{\mu}(\tau)-\vec{\mu}^\ast||^2 - \frac{(q(\vec{\mu}(\tau))-q(\vec{\mu}^\ast))^2}{\hat{g}^2}, \label{eq:ineq} \end{eqnarray}
where $\hat{g}^2$ is an upper bound of $||\vec{g}(\tau)||^2$. Since the second term on the right-hand-side of (\ref{eq:ineq}) is non-negative, it follows that $\lim_{\tau\rightarrow\infty}q(\vec{\mu}(\tau)) =q(\vec{\mu}^\ast)$. 

Summing Inequality (\ref{eq:ineq}) 
over $\tau$, we have 
\begin{eqnarray}\label{eq:SumBound}
  \sum_{\tau=1}^{\infty}(q(\vec{\mu}(\tau))-q(\vec{\mu}^\ast))^2  
    \le \hat{g}^2 ||\vec{\mu}(1)-\vec{\mu}^\ast||^2. \nonumber
\end{eqnarray}
That is, the square error sum is upper bounded. 
\end{proof}

\begin{theorem} \label{th:2}
The sequence $\{ q(\vec{\mu}(\tau)) \}$ converges faster than $\{ 1/\sqrt{\tau} \}$ as $\tau$ goes to infinity.
\end{theorem}
\begin{proof} 
Assume $\lim_{\tau\rightarrow\infty} \sqrt{\tau} ( q(\vec{\mu}(\tau)) - q(\vec{\mu}^\ast) )>0$. Then there is a sufficiently large $\tau'$ and a positive number $\xi$ such that $\sqrt{\tau}(q(\vec{\mu}(\tau))-q(\vec{\mu}^\ast))\ge \xi$, for all $\tau \ge \tau'$. 
Taking the square sum from $\tau'$ to $\infty$, we have:
\begin{eqnarray}\label{eq:SumUnbnd}
\sum_{\tau=\tau'}^{\infty}(q(\vec{\mu}(\tau))-q(\vec{\mu}^\ast))^2\ge \xi^2\sum_{\tau=\tau'}^{\infty}\frac{1}{\tau}=\infty. 
\end{eqnarray}
Eq.~(\ref{eq:SumUnbnd}) contradicts with Theorem~\ref{th:1}, which states that $\sum_{\tau=1}^{\infty}(q(\vec{\mu}(\tau))-q(\vec{\mu}^\ast))^2$ is bounded. Therefore, we have
\begin{eqnarray}\label{eq:CovgRate}
\lim_{\tau\rightarrow\infty}\frac{q(\vec{\mu}(\tau))-q(\vec{\mu}^\ast)}{1/\sqrt{\tau}}=0,
\end{eqnarray}
indicating that the convergence speed of $q(\vec{\mu}(\tau))$ is faster than that of $1/\sqrt{\tau}$.  
\end{proof}

\subsubsection{Case of Multiple Channels with Channel Bonding \label{subsec:mccbd}}

When there are multiple licensed channels, we first consider the case where the channel bonding/aggregation techniques are used by the transmitters and CR users~\cite{Corderio06, Mahmoud09}. With channel bonding, a transmitter can utilize all the available channels in $\mathcal{A}(t)$ collectively to transmit the mixed signal. We assume that at the end of the sensing phase in each time slot, CR users tune their SDR transceiver to the common control channel to receive the set of available channels $\mathcal{A}(t)$ from the BS.  Then each CR user can receive from all the channels in $\mathcal{A}(t)$ and decode its desired signal from the compound signal it receives.

This case is similar to the case of a single licensed channel. Now all the active CR users receive from the set of available channels $\mathcal{A}(t)$. We thus have $b_j^m=1$, for $m \in \mathcal{A}(t)$, and $b_j^m=0$, for $m \notin \mathcal{A}(t)$, $j=1,2,\cdots,N$. When all the $b_j^m$'s are determined this way, problem~(\ref{eq:ObjFun1}) is reduced to a nonlinear programming problem with constraints (\ref{eq:OrthSig}) and (\ref{eq:MaxPow}). The distributed algorithm described in Section~\ref{subsec:onec} can be applied to solve this reduced problem to get optimal solutions.

\subsubsection{Case of Multiple Channels without Channel Bonding \label{subsec:mcnocbd}}

We finally consider the case of multiple channels without channel bonding, where each CR user has a narrow band SDR transceiver and can only receive from one of the channels. We first present a greedy algorithm that leverages the optimal algorithm in Table~\ref{tab:SingleChan} for near-optimal solutions, and then derive a lower bound for its performance. 

\paragraph{Greedy Algorithm \label{subsubsec:gd}}
When $M>1$, the optimal solution to problem~(\ref{eq:ObjFun1}) depends also on the 
binary variables $b_j^m$'s, which determines whether user $j$ receives from channel $m$. Recall that there are two constraints for the $b_j^m$'s: (i) each user can use at most one channel (see~(\ref{eq:OneChan2})); (ii) the number of users on the same channel cannot exceed the number of transmitters $K$ (see~(\ref{eq:MaxSender})). 
Let $\vec{b}$ be the channel allocation vector with elements $b_j^m$'s, and $\Phi(\vec{b})$ the corresponding objective value for a given user channel allocation $\vec{b}$. 

We take a two-step approach to solve problem~(\ref{eq:ObjFun1}). 
First, we apply the greedy algorithm in Table~\ref{tab:MultiChan} to choose one available channel in $\mathcal{A}(t)$ for each CR user (i.e., to determine $\vec{b}$). Second, we apply the algorithm in Table~\ref{tab:SingleChan} to obtain a near-optimal solution for the given channel allocation $\vec{b}$.  

\begin{table} [!t]
\begin{center}
\caption{Channel Selection Algorithm for the Case of Multiple Channels without Channel Bonding}
\begin{tabular}{ll}
\hline
1: & Initialize $\vec{b}$ to a zero vector, user set $\mathcal{U}=\{1,\cdots,N\}$ \\
   & and user-channel set $\mathcal{C}=\mathcal{U}\times \mathcal{A}(t)$; \\
2: & WHILE ($\mathcal{C} \neq \emptyset$) \\
3: & $\;\;\;$ Find the user-channel pair $\{j',m'\}$, such that \\
   & $\;\;\;\;\;\;$ $\{j',m'\}=\arg\max_{\{(j,m) \in \mathcal{C}\}} \{\Phi(\vec{b}+\vec{\upsilon}_j^m)-\Phi(\vec{b})\}$; \\
4: & $\;\;\;$ Set $\vec{b}=\vec{b}+\vec{\upsilon}_{j'}^{m'}$ and remove $j'$ from $\mathcal{U}$; \\
5: & $\;\;\;$ IF ($\sum_{j=1}^N b_j^{m'}=K$)  \\
6: & $\;\;\;\;\;\;$ Remove $m'$ from $\mathcal{A}(t)$; \\
7: & $\;\;\;$ END IF \\
8: & $\;\;\;$ Update user-channel set $\mathcal{C}=\mathcal{U}\times \mathcal{A}(t)$; \\
9: & END WHILE \\
\hline
\end{tabular}
\label{tab:MultiChan}
\end{center}
\end{table}

In Table~\ref{tab:MultiChan}, $\vec{\upsilon}_j^m$ is a unit vector with 1 for the $[(j-1)\times M+m]$-th element and $0$ for all other elements, and $\vec{b}=\vec{b}+\vec{\upsilon}_{j'}^{m'}$ indicates choosing channel $m'$ for user $j'$. In each iteration, the user-channel pair $(j', m')$ that can achieve the largest increase in the objective value is chosen, as in Step 3. The complexity of the greedy algorithm in the worst case is $O(K^2 M^2)$.

\paragraph{Performance Bound \label{subsubsec:bd}}
We next analyze the greedy algorithm and derive a lower bound for its performance. Let $\nu_l$ be the sequence from the first 
to the $l$th user-channel pair selected by the greedy algorithm. The increase in objective value is denoted as:
\begin{eqnarray}\label{eq:DefFl}
F_l:=F(\nu_l,\nu_{l-1})=\Phi(\nu_l)-\Phi(\nu_{l-1}).
\end{eqnarray}
Sum up (\ref{eq:DefFl}) from 1 to $L$. We have $\sum_{l=1}^L F_l=\Phi(\nu_L)$ since $\Phi(\nu_0)=0$. Let $\Omega$ be the global optimal solution for user-channel allocation. 
Define $\pi_l$ as a subset of $\Omega$. For given $\nu_l$, $\pi_l$ is the subset of user-channel pairs that cannot be allocated due to the conflict with the $l$-th user channel allocation $\nu_l$ (but not conflict with the user-channel allocations in $\nu_{l-1}$).

\begin{lemma} \label{lemma:5}
Assume the greedy algorithm stops in $L$ steps, we have 
$$\Phi(\Omega)\le\Phi(\nu_L)+\sum_{l=1}^L\sum_{\sigma\in \pi_l} F(\sigma\cup\nu_{l-1},\nu_{l-1}).
$$
\end{lemma}
\begin{proof} 
The proof is similar to the proof of Lemma 7 in~\cite{Hu12JSAC} and is omitted for brevity. 
\end{proof}

\begin{theorem} The greedy algorithm for channel selection in Table~\ref{tab:MultiChan} can achieve an objective value that is at least $\frac{1}{|\mathcal{A}(t)|}$ of the global optimum in each time slot.
\end{theorem}
\begin{proof} 
According to Lemma~\ref{lemma:5}, it follows that:
\begin{eqnarray}\label{eq:LowBndIneq}
&& \hspace{-0.0in} \Phi(\Omega) \le \Phi(\nu_L) \hspace{-0.025in} + \hspace{-0.025in} \sum_{l=1}^L |\pi_l| F_l 
\le \Phi(\nu_L) \hspace{-0.025in} + \hspace{-0.025in} (|\mathcal{A}(t)| \hspace{-0.025in} - \hspace{-0.025in} 1)\sum_{l=1}^LF_l 
= |\mathcal{A}(t)|\Phi(\nu_L).
\end{eqnarray}
The second inequality is due to the fact that each user can choose at most one channel and there are at most $(|\mathcal{A}(t)|-1)$ pairs in $\pi_l$ according to the definition. The equality in~(\ref{eq:LowBndIneq}) is because $\sum_{l=1}^L F_l=\Phi(\nu_L)$. Then we have:
\begin{eqnarray}\label{eq:LUBnd}
\frac{1}{|\mathcal{A}(t)|}\Phi(\Omega)\le \Phi(\nu_L) \le \Phi(\Omega).
\end{eqnarray}
The greedy heuristic solution is lower bounded by $1/|\mathcal{A}(t)|$ of the global optimum. 
\end{proof}

Define competitive ratio $\chi = \Phi(\nu_L) / \Phi(\Omega) = 1/|\mathcal{A}(t)|$.  Assume all the licensed channels have identical utilization $\eta$. Since $|\mathcal{A}(t)|$ is a random variable, 
we take the expectation of $\chi$ and obtain:
\begin{eqnarray}\label{eq:AvgZeta}
\mathbb{E} [\chi]= \eta^M+\sum_{n=1}^M \left( \frac{1}{n} \right) \eta^{M-n}(1-\eta)^n. 
\end{eqnarray}

In Fig.~\ref{fig:RatioVsEta}, we evaluate the impact of channel utilization $\eta$ and the number of licensed channels $M$ on the competitive ratio.  We increase $\eta$ from $0.05$ to $0.95$ in steps of $0.05$ and increase $M$ from $6$ to $12$ in steps of $2$. 
The lower bound~(\ref{eq:LUBnd}) becomes tighter when $\eta$ is larger or when $M$ is smaller. For example, when $\eta=0.6$ and $M=6$, the greedy algorithm solution is guaranteed to be no less than 52.7\% of the global optimal. when $\eta$ is increased to 0.95, the greedy algorithm solution is guaranteed to be no less than 98.3\% of the global optimal. 

\begin{figure} [!t]
  \centering
  \includegraphics[width=4.5in, height=3.0in]{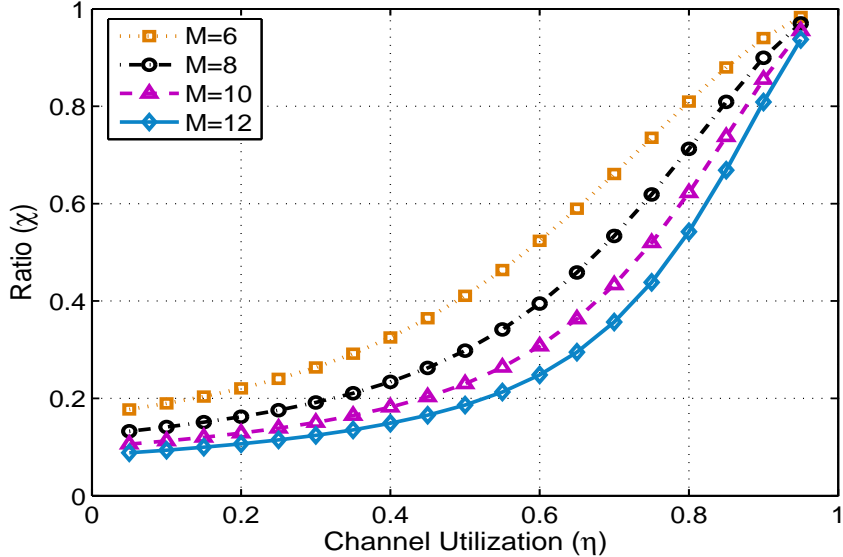}
  \caption{Competitive ratio $\mathbb{E} [\chi]$ defined in (\ref{eq:AvgZeta}) versus channel utilization $\eta$.}
  \label{fig:RatioVsEta}
\end{figure}

\subsection{Performance Evaluation \label{sec:sim2}}

We evaluate the performance of the proposed algorithms with a MATLAB implementation and the JVSM 9.13 Video Codec. We present simulation results for the following two scenarios: (i) a single licensed channel and (ii) multiple licensed channels without channel bonding, since we observe similar performance for the case of multiple licensed channels with channel bonding.  For comparison purpose, we also developed two simpler heuristic schemes that do not incorporate interference alignment.
\begin{itemize}
  \item 
  {\em Heuristic 1}: each CR user selects the best channel in $\mathcal{A}(t)$ 
  based on channel condition. The time slot is equally divided among the active
  users receiving from the same channel, to send their signals separately in each 
  time slice.
  \item 
  {\em Heuristic 2}: in each time slot, the active user with the best channel is
  selected for each available channel. The entire time slot is used to transmit this
  user's signal. 
\end{itemize}

\subsubsection{Case of a Single Licensed Channel}

In the first scenario, there are $K=4$ transmitters, i.e., one BS and three RNs. The channel utilization $\eta$ is set to 0.6 and the maximum allowable collision probability $\gamma$ is set to 0.2. There are three active CR users, each receives an MGS video stream from the BS: {\em Bus} to CR user 1, {\em Mobile} to CR user 2, and {\em Harbor} to CR user 3. The video sequences are in the Common Intermediate Format (CIF, 252$\times$288).  The GOP size of the videos is 16 and the delivery deadline $T$ is 10. The false alarm probability is $\epsilon_l^m=0.3$ and the miss detection probability is $\delta_l^m=0.3$ for all spectrum sensors. The channel bandwidth $B$ is 1 MHz. The peak power limit is 10 W for all the transmitters, unless otherwise specified.

We first plot the average Y-PSNRs of the three reconstructed MGS videos in Fig.~\ref{fig:OneChanPSNR}, i.e., only the Y (Luminance) component of the original and reconstructed videos are used. Among three schemes, the proposed algorithm achieves the highest PSNR value, while the two heuristic algorithms have similar performance. Note that the proposed algorithm is optimal in the single channel case. It achieves significant improvements ranging from 3.1 dB to 5.25 dB over the two heuristic algorithms. Such PSNR gains are considerable, since in video coding and communications, a half dB gain is distinguishable and worth pursing.

\begin{figure} [!t]
\centering
\includegraphics[width=4.5in, height=3.0in]{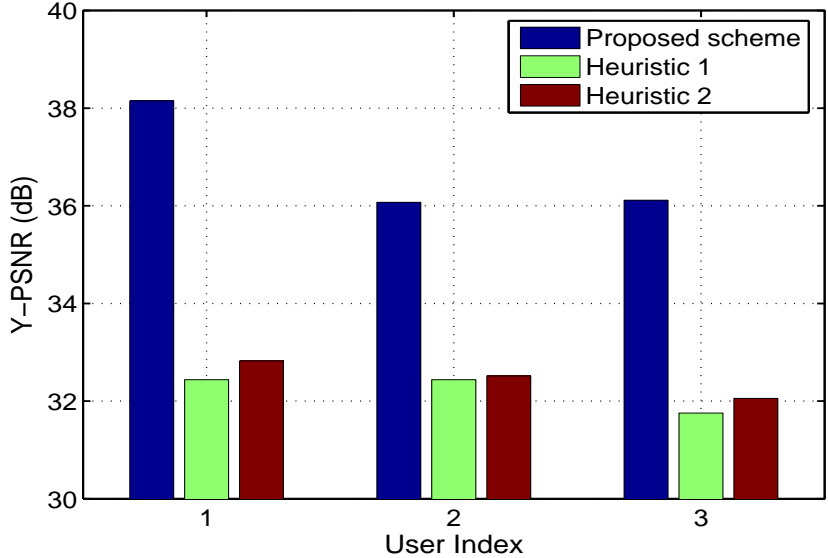}
\caption{Received video quality for each CR user with a single channel.}
\label{fig:OneChanPSNR}
\end{figure}

We next examine the convergence rate of the distributed algorithm. According to Theorem~\ref{th:2}, the distributed algorithm converges at a speed faster than $1/\sqrt{\tau}$ asymptotically. We compare the optimality gap of the proposed algorithm, i.e., $|q(\tau)-q^{*}|$, with series $10/\sqrt{\tau}$ in Fig.~\ref{fig:CovgRate}. Both curves converge to 0 as $\tau$ goes to infinity. It can be seen that the convergence speed, i.e., the slope of the curve, of the proposed scheme is larger than that of $10/\sqrt{\tau}$ after about $10$ iterations.  
The convergence of the optimality gap is much faster than $10/\sqrt{\tau}$, which exhibits a heavy tail. 

\begin{figure} [!t]
\centering
\includegraphics[width=4.5in, height=3.0in]{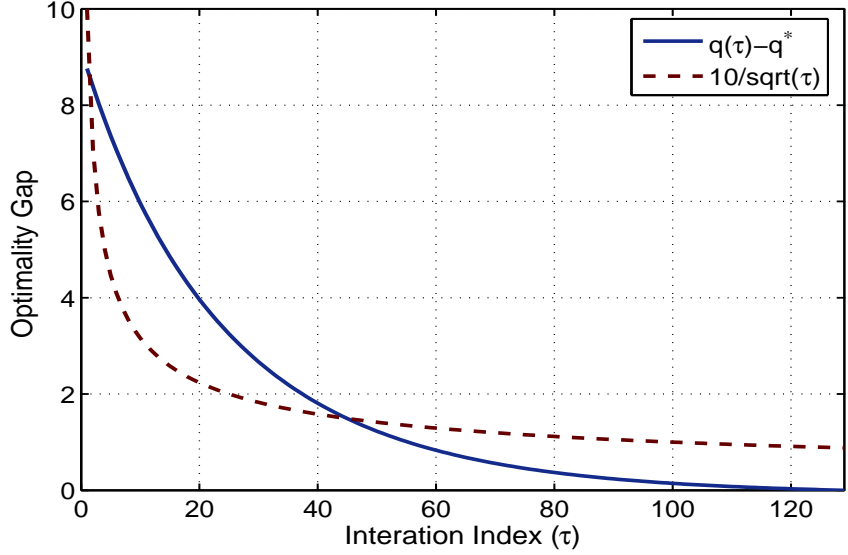}
\caption{Convergence rate of the distributed algorithm with a single channel.}
\label{fig:CovgRate}
\end{figure}

In the case of multiple channels with channel bonding, the performance of the proposed algorithm is similar to that in the single channel case. We omit the results for lack of space.

\subsubsection{Case of Multiple Channels without Channel Bonding}

We next investigate the second scenario with six licensed channels and four transmitters. There are 12 CR users, each streaming one of the three different videos {\em Bus}, {\em Mobile}, and {\em Harbor}. 
The rest of the parameters are the same as those in the single channel case, unless otherwise specified. Eq.~(\ref{eq:LowBndIneq}) can also be interpreted as an upper bound on the global optimal, i.e., $\Phi(\Omega)\le |\mathcal{A}(t)| \Phi(\nu_L)$, which is also plotted in the figures. Each point in the following figures is the average of 10 simulation runs with different random seeds. The 95\% confidence intervals are plotted as error bars, which are generally negligible.

The impact of channel utilization $\eta$ on received video quality is presented in Fig.~\ref{fig:MultChanEta}. We increase $\eta$ from $0.3$ to $0.9$ in steps of $0.15$, and plot the Y-PSNRs of reconstructed videos averaged over all the 12 CR users. Intuitively, a smaller $\eta$ allows more transmission opportunities for CR users, thus allowing the CR users to achieve higher video rates and better video quality. This is shown in the figure, in which all four curves decrease as $\eta$ is increased. We also observe that the gap between the upper bound and proposed schemes becomes smaller as $\eta$ gets larger, from 32.65 dB when $\eta=0.3$ to 0.63 dB when $\eta=0.9$. This trend is also demonstrated in Fig.~\ref{fig:RatioVsEta}. The proposed scheme outperforms the two heuristic schemes with considerable gains, ranging from 0.8 dB to 3.65 dB. 

\begin{figure} [!t]
\centering
\includegraphics[width=4.5in, height=3.0in]{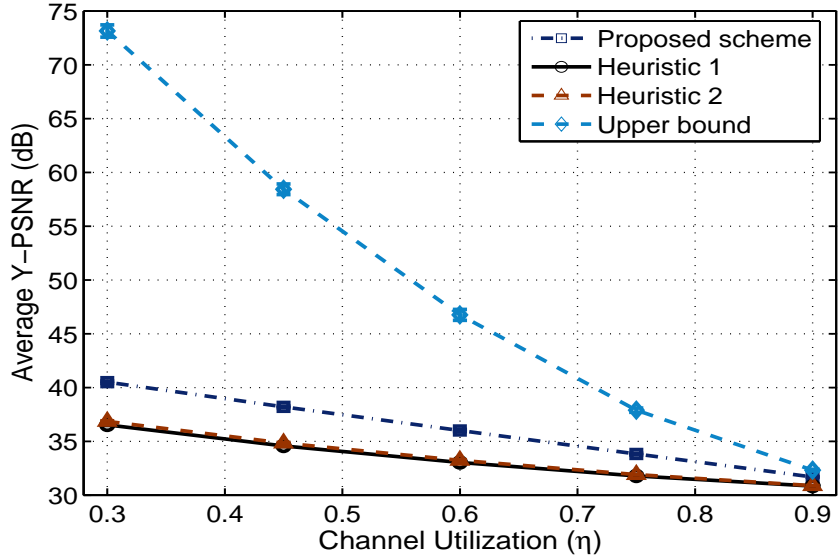}
\caption{Reconstructed video quality vs. channel utilization $\eta$ in the multi-channel without channel bonding case.}
\label{fig:MultChanEta}
\end{figure}

Finally, we investigate the impact of the number of transmitters $K$ on the video quality. In this simulation we increase $K$ from 2 to 6 with step size 1. The average Y-PSNRs of all the 12 CR users are plotted in Fig.~\ref{fig:MultChanK}. As expected, the more transmitters, the more effective the interference alignment technique, and thus the better the video quality. The proposed algorithm achieves gains ranging from 1.78 dB (when $K=2$) to 4.55 dB (when $K=6$) over the two heuristic schemes. 

\begin{figure} [!t]
\centering
\includegraphics[width=4.5in, height=3.0in]{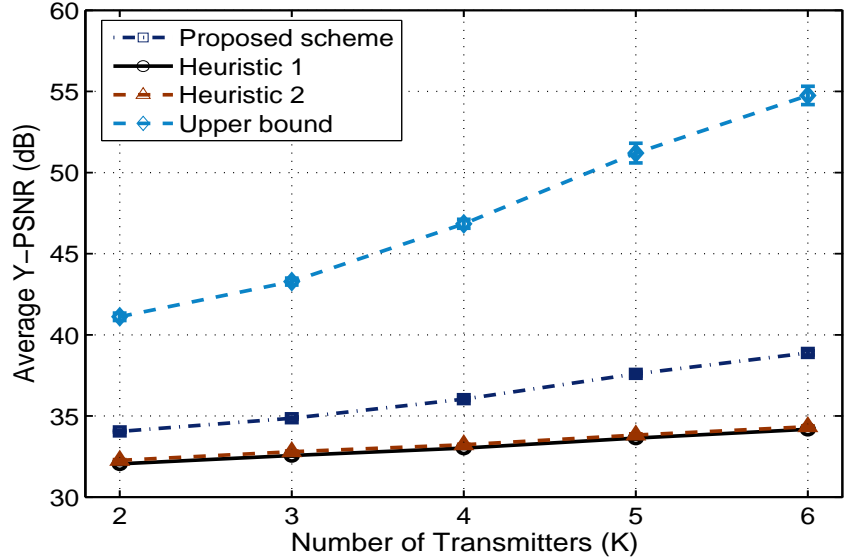}
\caption{Reconstructed video quality vs. number of transmitters $K$ in the multi-channel without channel bonding case.}
\label{fig:MultChanK}
\end{figure}

\section{Conclusions}\label{sec:coop_conc}
In this paper, we first studied the problem of cooperative relay in CR networks.  We modeled the two cooperative relay strategies, i.e., DF and AF, which are integrated with $p$-Persistent CSMA.  We 
analyzed their throughput performance and compared them under various parameter ranges.  
Cross-point with the AF and DF curves are found when some parameter is varied, indicating that each of them performs better in a certain parameter range; there is no case of dominance for the two strategies.  
Considerable gains were observed over conventional DL transmissions, as achieved by exploiting cooperative diversity with the cooperative relays in CR networks. 

Then, we investigated the problem of interference alignment for MGS video streaming in a cooperative relay enhanced CR network. We presented a stochastic programming formation, and derived a reformulation that leads to considerable reduction in computational complexity. A distributed optimal algorithm was developed for the case of a single channel and the case of multi-channel with channel bonding, with proven convergence and convergence speed. We also presented a greedy algorithm for the multi-channel without channel bonding case, with a proven performance bound. The proposed algorithms are evaluated with simulations and are shown to outperform two heuristic schemes without interference alignment with considerable gains.

\bibliographystyle{IEEEtran}
\bibliography{cr_video_femto,MyWork}

\end{document}